\documentclass[journal]{IEEEtran}
\usepackage{amsmath}
\usepackage{bbm}
\usepackage{graphicx}
\usepackage{epsfig}
\usepackage{array}
\usepackage{psfrag}
\usepackage{amssymb}
\usepackage{mathdots}
\usepackage{color}
\usepackage{pstricks,pst-node,pst-text,pst-3d,pst-plot}
\usepackage{psfrag}
\usepackage{enumerate}
\usepackage{url,cite}
\usepackage{amsfonts}%
\usepackage{amsthm}
\usepackage{dsfont}%
\usepackage{verbatim}%
\usepackage{setspace}
\usepackage{float}
\usepackage{url}
\usepackage{fancyhdr}
\usepackage{bm}
\usepackage{lipsum} %\footnote

\usepackage{algorithm}
\usepackage{algorithmic}
\newtheorem{thm}{Theorem}
\newtheorem{lma}{Lemma}
\newtheorem{Def}{Definition}
\DeclareMathOperator{\E}{\mathbb{E}}
\newcommand{\bP}{\overline{P}}

\newcommand{\lb}{\left (}
\newcommand{\rb}{\right )}

\newcommand{\script}[1]{{\mathcal {#1}}}

\newcommand{\Pmax}{P_{\rm max}}
\newcommand{\Pmin}{P_{\rm min}}

\newcommand{\Iinst}{I_{\rm inst}}
\newcommand{\Iavg}{I_{\rm avg}}
\newcommand{\gmax}{g_{\rm max}}
\newcommand{\fgammai}{f_{\gamma_i}}
\newcommand{\fgi}{f_{g_i}}
\newcommand{\EE}[1]{\E \left[ #1 \right]}
\newcommand{\EEU}[1]{\E_{\bfU(k)} \left[ #1 \right]}
\newcommand{\EEY}[1]{\E_{\bfU(k)} \left[ #1 \right]}
\newcommand{\Prob}[1]{{\bf Pr}\left[ #1 \right]}
\newcommand{\bgi}{\overline{g}_i}
\newcommand{\bW}{\overline{W}}

\newcommand{\bfP}{{\bf P}}
\newcommand{\bfPst}{{\bf P}^*}

\newcommand{\bfpi}{{\bm{\pi}}}
\newcommand{\bfpist}{{\bm{\pi}}^*}

\newcommand{\bgamma}{\overline{\gamma}}

\newcommand{\Xvq}{\{X(k)\}_{k=0}^\infty}
\newcommand{\Yivq}{\{Y_i(k)\}_{k=0}^\infty}

\newcommand{\bfY}{{\bf Y}}
\newcommand{\bfU}{{\bf U}}
\newcommand{\bfr}{{\bf r}}

\newcommand{\parFdef}[1]{\triangleq [#1_1(k),\cdots,#1_N(k)]^T}
\newcommand{\DOIC}{\emph{DOIC}}
\newcommand{\DOAC}{\emph{DOAC}}
\newcommand{\DOACopt}{DOAC-Pow-Alloc}
\newcommand{\brho}{\overline{\rho}}
\newcommand{\Ri}{R_i^{(t)}}
\newcommand{\FDurK}{T_k}
\newcommand{\Wup}{W^{\rm up}_{\pi_j}}

\newcommand{\gammaerr}{\gamma_i^{\rm obs}(t)}
\newcommand{\gerr}{g_i^{\rm obs}(t)}
\newcommand{\dHigh}{60}
\newcommand{\dLow}{45}
\newcommand{\Obj}{\Psi}

\newcommand{\psiI}{\psi_{\pi_j}^{\rm I}}
\newcommand{\psiD}{\psi_{\pi_j}^{\rm D}}
\newcommand{\psiInt}{\psi^{\rm I}}
\newcommand{\psiDel}{\psi^{\rm D}}

\newcommand{\tP}{\tilde{P}_{\pi_j}}
\newcommand{\tPsi}{\tilde{\Psi}}

\newcommand{\brhomax}{\brho^{\rm max}}

\newcommand{\sB}{s^{\rm B}}
\newcommand{\sNB}{s^{\rm NB}}

\newcommand{\Rmax}{R_{\rm max}}
\newcommand{\gammamax}{\gamma_{\rm max}}

\newcommand{\Ts}{T}

\newcommand{\Pit}{P_i^{(t)}}
\newcommand{\git}{g_i^{(t)}}
\newcommand{\HOL}{L^{\rm rem}}

\newcommand{\SU}[1]{{\rm SU}_{#1}}
\newcommand{\ProbA}{\eqref{Problem}}
\newcommand{\ProbB}{\eqref{Prob}}
\newcommand{\trho}{\tilde{\rho}}
\newcommand{\tS}{{\bf \tilde{S}}}
\newcommand{\sS}{\script{S}}

\ifCLASSINFOpdf
\else
\fi

\begin{document}
\title{Joint Scheduling and Power-Control for Delay Guarantees in Heterogeneous Cognitive Radios}

\author{Ahmed~Ewaisha
        and~Cihan~Tepedelelio\u{g}lu\\
				School of Electrical, Computer, and Energy Engineering, Arizona State University, USA\\
		Email:\{ewaisha, cihan\}@asu.edu% <-this % stops a space
\thanks{The authors are with the School of Electrical, Computer and Energy Engineering, Arizona State University, Tempe, Az, 85287 USA.% e-mail: (see http://www.public.asu.edu/~aewaisha/ and wireless.faculty.asu.edu/).
}% <-this % stops a space
\thanks{The work in this paper has been partially supported by NSF Grant CCF-1117041.}% <-this % stops a space
\thanks{Parts of this work appeared in the 2015 Asilomar Conference on Signals, Systems, and Computers \cite{Ewaisha_Asilomar_2015}.}
}




\maketitle

\begin{abstract}
An uplink multi secondary user (SU) cognitive radio system having average delay constraints as well as an interference constraint to the primary user (PU) is considered. If the interference channels between the SUs and the PU are statistically heterogeneous due to the different physical locations of the different SUs, the SUs will experience different delay performances. This is because SUs located closer to the PU transmit with lower power levels. Two dynamic scheduling-and-power-allocation policies that can provide the required average delay guarantees to all SUs irrespective of their locations are proposed. The first policy solves the problem when the interference constraint is an instantaneous one, while the second is for problems with long-term average interference constraints. We show that although the average interference problem is an extension to the instantaneous interference one, the solution is totally different. The two policies, derived using the Lyapunov optimization technique, are shown to be asymptotically delay optimal while satisfying the delay and interference constraints. Our findings are supported by extensive system simulations and shown to outperform existing policies as well as shown to be robust to channel estimation errors.
\end{abstract}


\section{Introduction}

The problem of scarcity in the radio spectrum has led to a wide interest in cognitive radio (CR) networks. CRs refer to devices that coexist with the licensed spectrum owners called the primary users (PUs). CRs are capable of dynamically adjusting their transmission parameters according to the environment to avoid harmful interference to the PUs. CR users adjust their transmission power levels, and their rates, according to the interference level the PUs can tolerate. However, this adjustment can be at the expense of quality of service (QoS) provided to the CR users, if not designed carefully.

In real-time applications, such as audio and video conference calls, one of the most effective QoS metrics is the average time a packet spends in the queue before being fully transmitted, quantified by average queuing delay. This is because as this amount of queuing delay increases, the user receiving the packet will have to wait for the packet until it is received. This causes intermittent streaming of the audio and video which is an undesirable feature of these applications. Hence, the average queuing delay needs to be as small as possible to prevent jitter and guarantee acceptable QoS for these applications \cite{shakkottai2002scheduling,kang2013performance}. Queuing delay has gained strong attention recently and scheduling algorithms have been proposed to guarantee small delay in wireless networks (see e.g., \cite{asadi2013survey} for a survey on scheduling algorithms in wireless systems). In \cite{li2011delay}, the authors study joint scheduling-and-power-allocation to minimize the delay in the presence of an average power constraint. A power allocation and routing algorithm is proposed in \cite{neely2003power} to maximize the capacity region under an instantaneous power constraint. In \cite{Scheduling_HSDPA} the authors propose a scheduling algorithm to maximize the cell throughput while maintaining a level of fairness between the users in the cell. In a two-queue setup, one with light traffic and one with light traffic, \cite{Two_Q_Light_Hvy} showed that giving priority to light traffic guarantees the best tail behavior of the delay distribution for both queues under on-off wireless channels.

Unfortunately, applying the existing scheduling algorithms to secondary users (SUs) in CR systems results in undesired delay performance. This is because SUs located physically closer to the PUs might suffer from larger delays because closer SUs transmit with smaller power levels. The SUs should be scheduled and have their power controlled in such a way that prevents harmful interference to the PUs since they share the same spectrum.

The problem of scheduling and/or power control for CR systems has been widely studied in the literature (see e.g., \cite{Letaief_PU_Known_Location,Iter_Bit_Allocation_OFDM,Neely_CNC_2009,6464638,6924778,6810836,NEP_Distributed,Ewaisha_TVT2015}, and the references therein). An uplink CR system is considered in \cite{Letaief_PU_Known_Location} where the authors propose a scheduling algorithm that minimizes the interference to the PU where all users' locations including the PU's are known to the secondary base station. The objective in \cite{6924778} is to maximize the total network's welfare. While this could give good performance in networks with users having statistically homogeneous channels, the users might experience degraded QoS when their channels are heterogeneous. Reference \cite{6810836} has considered users with heterogeneous throughput requirements. This model can be applied best for regular non-real-time applications. While for real time applications, the secondary users might suffer high delays even if their throughput was optimum. In \cite{NEP_Distributed} a distributed scheduling algorithm that uses an on-off rate adaptation scheme is proposed. The authors of \cite{Ewaisha_Throughput_Maximization} and \cite{Ewaisha_TVT2015} propose a closed-form water-filling-like power allocation policy to maximize the CR system's per-user throughput. %Although multiple SUs were considered, the scheduling scheme neglects the SUs' queuing delay.
The work in \cite{Neely_CNC_2009} proposes a scheduling algorithm to maximize the capacity region subject to a collision constraint on the PUs. The algorithms proposed in all these works aim at optimizing the throughput for the SUs while protecting the PUs from interference. However, providing guarantees on the queuing delay in CR systems was not the goal of these works. 

%Although queuing theory, that was originally developed to study packets at a server, can be applied to wireless channels, the challenges are different.
The fading nature of the wireless channel requires adapting the user's power and rate according to the channel's fading coefficient. Many existing works on scheduling algorithms consider two-state on-off wireless channels and do not consider multiple fading levels. Among the relevant references that consider a more general fading channel model are \cite{neely2003power} and \cite{Min_Pow_4_Delay_NonCR} which do not include an average interference constraint, as well as \cite{E_Hossain_CR_Delay_Analysis,Fading_No_Scheduling} where the optimization over the scheduling algorithm was not considered.

From a technical point of view, the closest to our work is \cite{li2011delay} which studies the joint scheduling-and-power-allocation problem, and assumes that all users process packets with the same power since it discusses the problem of processing jobs at a CPU. The CPU problem considered in \cite{li2011delay} is a special case of the wireless channel problem herein. Finally, the problem is formulated in continuous time in \cite{li2011delay} where the packet service time follows a continuous time distribution that is easier to analyze than discrete ones. In wireless settings, the fading coherence time provides a naturally discrete/slotted framework which brings with it its own combinatorial technical challenges.\\
Unlike \cite{SU_Selection_Energy_Minimisation} that studies the effect of heterogeneity among SUs on the detection of the PU, in this paper, we study the effect of this heterogeneity on the delay performance of SUs. We consider the joint scheduling and power control problem of minimizing the sum average delay of SUs subject to interference constraints at the PU, for the first time in the literature. Our model relaxes the equal transmission power constraint among SUs. Moreover, our algorithm provides per-user average delay guarantees so that each SU meets its delay requirements. We consider both instantaneous and average interference constraints. The technical challenge of this problem lies in its objective function which is the sum of average delays. This objective is not a simple function in the users' power levels thus making the joint optimization problem at hand challenging. Moreover, the power allocation policy needs to protect the PU from interference. The novel contributions of this paper include: i) proposing two joint-power-control-and-scheduling policies that are optimal with respect to the sum of average delays of SUs, a policy for the problem under instantaneous interference constraint and the other under average interference constraint; ii) exploiting the unique structure of the problem to provide an optimal power allocation algorithm of a lower complexity than exhaustive search; iii) using Lyapunov analysis to show that the policy meets the heterogeneous per-user average delay requirements; iv) proposing an alternative low-complexity suboptimal policy that is shown to have a near-to-optimal performance with polynomial complexity in the number of SUs.% since it is equivalent to that of sorting a vector of $N$ numbers, where $N$ is the number of SUs in the system.% We show that conventional existing algorithms as the max-weight scheduling algorithm, if applied directly, can degrade the quality of service of both SUs as well as the PUs.

%Despite this possible degradation in performance, SUs might or might not suffer depending on the type of data being transmitted. For example, assuming the QoS metric to be the average waiting time in the queue, SUs might still not suffer even with this degradation in performance if they were transmitting delay-insensitive data. On the other hand, if the data was delay-sensitive, then the need of algorithms that guarantees the SUs' delay performance is inevitable.

% We present to different scenarios: 1) Instantaneous interference constraint, and 2) Average interference constraint. In the former, the PUs' interference constraint is such that the instantaneous interference does not exceed a prespecified threshold, while in the latter, the interference, averaged over an infinite horizon, is bounded by some constant.

The rest of the paper is organized as follows. The network model and the underlying assumptions are presented in Section \ref{Model}. In Section \ref{Prob_Statement} we formulate the problem mathematically for both the instantaneous as well as the average interference constraints. The proposed policies for both scenarios, their optimality and complexity are presented in Section \ref{Proposed_Algorithm} as well as an alternative suboptimal policy. Section \ref{Results} presents our extensive simulation results. The paper is concluded in Section \ref{Conclusion}.

In this manuscript, we use bold to indicate vectors ${\bf X}$, and calligraphic font to indicate sets $\script{X}$. All logarithms are to the natural base $e$. We use $x^+$ to indicate $\max(x,0)$, $x^*$ to indicate the optimum power of $x$, $\vert \script{X}\vert$ for the cardinality of the set $\script{X}$, $\EE{\cdot}$ to indicate the expected value and $\E_{\bf X}\left[\cdot\right]$ for the expectation conditioned on the random vector ${\bf X}$.

\section{System Model}
\label{Model}
We assume a CR system consisting of a single secondary base station (BS) serving $N$ secondary users (SUs) indexed by the set $\script{N}\triangleq \{1,\cdots N\}$ (Fig. \ref{Cell_Fig}). We are considering the uplink phase where each SU has its own queue buffer for packets that need to be sent to the BS. The SUs share a single frequency channel with a single PU that has licensed access to this channel. The CR system operates in an underlay fashion where the PU is using the channel continuously at all times. SUs are allowed to transmit as long as they do not cause harmful interference to the PU. In this work, we consider two different scenarios where the interference can be considered as harmful. The first is an instantaneous interference constraint where the interference received by the PU at any given slot should not exceed a prespecified threshold $\Iinst$, while the second is an average interference constraint where the interference received by the PU averaged over a large duration of time should not exceed a prespecified threshold $\Iavg$. Moreover, in order for the secondary BS to be able to decode the received signal, no more than one SU at a time slot is to be assigned the channel for transmission.

\begin{figure}%
\centering
\includegraphics[width=0.7\columnwidth]{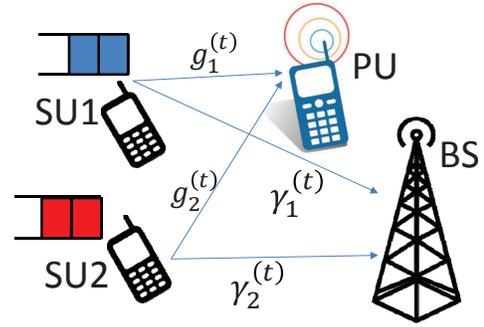}%
\caption{The CR system considered is an uplink one with $N$ SUs (in this figure $N=2$) communicating with their BS. There exists an interference link between each SU and the existing PU. The PU is assumed to be using the channel continuously.}%
\label{Cell_Fig}%
\end{figure}

\subsection{Channel and Interference Model}
We assume a time slotted structure where each slot is of duration $\Ts$ seconds, and equal to the coherence time of the channel. The channel between $\SU{i}$ and the BS is block fading, that is, the instantaneous power gain $\gamma_i^{(t)}$, at time slot $t$, is fixed within the time slot and changes independently in the following time slot. We assume that $\gamma_i^{(t)}$ follows the probability mass function $\fgammai(\gamma)$ with mean $\bgamma_i$ and independent and identically distributed (i.i.d.) across time slots, and $\gammamax$ is the maximum gain that $\gamma_i^{(t)}$ could take. The channel gain is also independent across SUs but not necessary identically distributed allowing heterogeneity among users%\footnote{The assumption of identical distributions can be relaxed but is assumed for presentation purposes only.}
. %We denote $\bfGamma^{(t)} \pardef{\gamma}$, and $\overline{\bfGamma} \triangleq [\bgamma_1,\cdots,\bgamma_N]^T$.
SUs use a rate adaptation scheme based on the channel gain $\gamma_i^{(t)}$. The transmission rate of $\SU{i}$ at time slot $t$ is
\begin{equation}
\Ri=\Ts\log \lb 1+P_i^{(t)}\gamma_i^{(t)} \rb \hspace{0.1in} \rm{bits},
\label{Tx_Rate}
\end{equation}
where $P_i^{(t)}$ is the power by which $\SU{i}$ transmits its bits at slot $t$. We assume that there exists a finite maximum rate $\Rmax$ that the SU cannot exceed. This rate is dictated by the maximum power $\Pmax$ and the maximum channel gain $\gammamax$.%We write $\mui$ as an explicit function of $P_i^{(t)}$ to indicate the dependency on the power policy used by $\SU{i}$, and not to indicate the dependency on the time index $t$ since $\mui$ is an average value that is not a function of time.

The PU experiences interference from the SUs through the channel between each SU and the PU. The interference channel between $\SU{i}$ and the PU, at slot $t$, has a power gain $g_i^{(t)}$ following the probability mass function $\fgi(g)$ with mean $\bgi$, and having $\gmax$ as the maximum value that $\git$ could take. These power gains are assumed to be independent among SUs but not identically distributed. We assume that $\SU{i}$ knows the value of $\gamma_i^{(t)}$ as well as $g_i^{(t)}$, at the beginning of slot $t$ through some channel estimation phase (see \cite[Section VI]{haykin2005cognitive}). Techniques to identify the modulation type can be found in references as \cite{Bari15asilomar2} which discusses the identification of PSK, 16-QAM and FM as well as \cite{Bari15asilomar3} for the continuous time FSK. The channel estimation to acquire $g_i^{(t)}$ can be done by overhearing the pilots transmitted by the primary receiver, when it is acting as a transmitter, to its intended transmitter \cite[Section VI]{haykin2005cognitive}. The channel estimation phase is out of the scope of this work, however the effect of channel estimation errors will be discussed in Section \ref{Results}.

\subsection{Queuing Model}
\subsubsection{Arrival Process} We assume that packets arrive to the $\SU{i}$'s buffer at the beginning of each slot. The number of packets arriving to $\SU{i}$'s buffer follows a Bernoulli process with a fixed parameter $\lambda_i$ packets per time slot. Following the literature, packets are buffered in infinite-sized buffers \cite[pp. 163]{Bertsekas_Data_Networks} and are served according to the first-come-first-serve discipline. Each packet has a fixed length of $L$ bits that is constant for all users. We note that the analysis of the random $L$ case \cite{Bertsekas_Data_Networks} would not be significantly different than the deterministic case, thus we discuss the fixed case for a better presentation of the paper. In this paper, we study the case where $L\gg\Rmax$ which is a typical case for packets with large sizes as video packets \cite{semiconductor2008long}. Due to the randomness in the channels, each packet takes a random number of time slots to be transmitted to the BS. This depends on the rate of transmission $\Ri$ as will be explained next.

\subsubsection{Service Process} When $\SU{i}$ is scheduled for transmission at slot $t$, it transmits $M_i^{(t)}$ bits of the head-of-line (HOL) packet of its queue. The remaining bits of this HOL packet remain in the HOL of $\SU{i}$'s queue until it is reassigned the channel in subsequent time slots. The values $M_i^{(t)}$ and $\HOL_i(t)$ are given by
\begin{align}
M_i^{(t)}&\triangleq\min \lb \Ri,\HOL_i(t) \rb \hspace{0.1in} \rm{bits, and}
\label{Num_Bits}\\
\HOL_i(t+1)&\triangleq \HOL_i(t) - M_i^{(t)},
\label{HOL}
\end{align}
respectively, where $\HOL_i(t)$ is the remaining number of bits of the HOL packet at $\SU{i}$ at the beginning of slot $t$. $\HOL_i(t)$ is initialized by $L$ whenever a packet joins the HOL position of $\SU{i}$'s queue so that it always satisfies $0\leq \HOL_i(t)\leq L$, $\forall t$. A packet is not considered transmitted unless all its $L$ bits are transmitted, i.e. unless $\HOL_i(t)$ becomes zero, at which point $\SU{i}$'s queue decreases by 1 packet. At the beginning of slot $t+1$ the following packet in the buffer, if any, becomes $\SU{i}$'s HOL packet and $\HOL_i(t+1)$ is reset back to $L$ bits. The $\SU{i}$'s queue evolves as follows
\begin{equation}
Q_i^{(t+1)}= \lb Q_i^{(t)} + \vert \script{A}_i^{(t)}\vert - S_i^{(t)} \rb^+,
\label{Queue}
\end{equation}
where $\script{A}_i^{(t)}$ is the set carrying the index of the packet, if any, arriving to $\SU{i}$ at slot $t$, thus $\vert \script{A}_i^{(t)}\vert$ is either $0$ or $1$ since at most one packet per slot can arrive to $\SU{i}$; %\footnote{Our model is also valid if the arrivals follow a discrete time process where $\script{A}_i^{(t)}$ will, then, represent the packets arrived at the beginning of slot $t$.}
the packet service indicator $S_i^{(t)}=1$ if $\HOL_i(t)$ becomes zero at slot $t$.

The service time $s_i$ of $\SU{i}$ is the number of time slots required to transmit one packet for $\SU{i}$, excluding the service interruptions. Using the assumption $L\gg \Rmax$ to approximate \eqref{Num_Bits} with $M_i^{(t)}=\Ri$, it can be shown that the average service time $\EE{s_i}=L/\EE{\Ri}$ time slots per packet where the expectation is taken over the channel gain $\gamma_i^{(t)}$ as well as over the power $P_i^{(t)}$ when it is channel dependent and random. One example of a random power policy is the \emph{channel inversion} policy as will be discussed later (see \eqref{Power_Allocation}). The service time is assumed to follow a general distribution throughout the paper that depends on the distribution of $\Pit\gamma_i^{(t)}$.

We define the delay $W_i^{(j)}$ of a packet $j$ as the total amount of time, in time slots, packet $j$ spends in $\SU{i}$'s buffer from the slot it joined the queue until the slot when its last bit is transmitted. %\footnote{We do not include the transmission slot when calculating $W_i^{(j)}$, but we include the residual time which is non-zero since packets are allowed to arrive in the middle of a time slot.}
The time-average delay experienced by $\SU{i}$'s packets is given by \cite{li2011delay}
%\begin{equation}
%\bW_i \triangleq \lim_{T \rightarrow \infty} \frac{\EE{\sum_{t=1}^T{\sum_{j\in\script{A}_i^{(t)}}{W_i^{(j)}}}}}{\EE{\sum_{t=1}^T{\vert \script{A}_i^{(t)}\vert}}}
%\label{Delay}
%\end{equation}
%which is the expected total amount of time spent by all packets arriving in a time interval, of a large duration, normalized by the expected number of packets that arrived in this interval.% $j$ arrives to $\SU{i}$'s buffer at some time during slot $k_1$, waits in the buffer until it becomes the head-of-line (HOL) packet at instant $t_2$, then is transmitted at time instant $t_3$ and takes a duration of $S_i(\gamma_i(j))$ units of time for transmission to the BS, at which it is assumed to have left the system. Hence, the total waiting time of this packet is $W_i(j)=t_3-t_1+S_i(\gamma_i(j))$ which has a mean $\bW_i(\SA)$ that depends on the scheduling algorithm $\SA$ as well as the power $P_i$. We will drop the argument $\SA$ whenever it is clear from the context.

\subsection{Transmission Process}
%At the beginning of frame $k$, the BS selects the priority list $\bfpi(k) \parFdef{\pi}$ where $\pi_j(k)$ is the index of the SU who will be assigned priority $j$ during frame $k$.
At the beginning of each time slot $t$, the BS schedules a SU and broadcasts its index $i^*$ and its power $P_{i^*}^{(t)}$ to all SUs on a common control channel. $\SU{i^*}$, in turn, begins transmission of $M_{i^*}^{(t)}$ bits of its HOL packet with a constant power $P_{i^*}^{(t)}$. We assume the BS receives these bits error-free by the end of slot $t$ then a new time slot $t+1$ starts. In this paper, our main goal is the selection of the $\SU{i^*}$ which is a scheduling problem, as well as the choice of the power $P_{i^*}^{(t)}$ which is power allocation. We now elaborate further on this problem.
%It will be shown that the proposed algorithm proposed in Section \ref{Proposed_Algorithm} can be implemented in a distributed fashion allowing the optimal power $P_{i^*}^{(t)}$ to be calculated locally.

\section{Problem Statement}
\label{Prob_Statement}
Each $\SU{i}$ has an average delay constraint $\bW_i \leq d_i$ that needs to be satisfied. Moreover, there are two types of interference constraints that the SU needs to meet in order to coexist with the PU. Before discussing both types and stating the problem associated with each one, we first give some definitions.

\subsection{Frame-Based Policy}
\label{Frame_Based_Policy}
In this work, we are interested in frame-based scheduling policies. The idea of dividing time into frames and assigning fixed scheduling and power allocation policy for each frame was also used in \cite{li2011delay}. We divide time into frames where frame $k$ consists of a random number $T_k$ time slots and update the power allocation and scheduling at the beginning of each frame. Where each frame begins and ends is specified by idle periods and will be precisely defined later in this section. During frame $k$, SUs are scheduled according to some priority list $\bfpi(k)$ and each SU is assigned some power to be used when it is assigned the channel. The priority list and the power functions are fixed during the entire frame $k$.

Define $\bfpi(k) \parFdef{\pi}$ where $\pi_j(k)$ is the index of the SU who is given the $j$th priority during frame $k$. Given $\bfpi(k)$, the scheduler becomes a priority scheduler with preemptive-resume priority queuing discipline \cite[pp. 205]{Bertsekas_Data_Networks}.%one can use a ``frame-based'' random priority list as the scheduling policy. This policy schedules the SUs based on a strict preemptive priority list and updates this list, at the beginning of each frame, randomly following some optimum, genie-aided, distribution defined over the $N!$ distinct priority lists. Hence, 

Frame $k$ consists of $T_k\triangleq\vert \script{F}(k)\vert$ consecutive time-slots, where $\script{F}(k)$ is the set containing the indices of the time slots belonging to frame $k$ (see Fig. \ref{Frame_Structure}). Each frame consists of exactly one \emph{idle period} followed by exactly one \emph{busy period}, both are defined next.

\begin{figure}
\centering
\includegraphics[width=0.8\columnwidth]{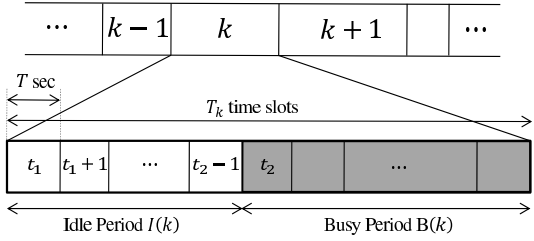}%
\caption{Time is divided into frames. Frame $k$ has $\FDurK\triangleq\vert \script{F}(k)\vert$ slots, each is of duration $\Ts$ seconds. Different frames can have different number of time slots.}%
\label{Frame_Structure}%
\end{figure}

\begin{Def}
\label{Idle_Def}
An idle period is the time interval formed by the consecutive time slots where all SUs have empty buffers. An idle period starts with the time slot $t_1$ following the completion of transmission of the last packet in the system, and ends with a time slot $t_2$ when one or more of the SUs' buffer receives one a new packet to be transmitted (see Fig. \ref{Frame_Structure}). In other words, $t_1$ satisfies $\sum_{i\in\script{N}}Q_i^{(t_1)}=0$ and $\sum_{i\in\script{N}}Q_i^{(t_1-1)}\neq 0$, while $t_2$ satisfies $\sum_{t=t_1}^{t_2-1}\sum_{i\in\script{N}}Q_i^{(t)}=0$ and $\sum_{i\in\script{N}}Q_i^{(t_2)}\neq 0$.
\end{Def}

\begin{Def}
\label{Busy_Def}
Busy period is the time interval between two consecutive idle periods.
\end{Def}
The duration of the idle period $I(k)$ and busy period $B(k)$ of frame $k$ are random variables, thus $T_k=I(k)+B(k)$ is random as well. Since frames do not overlap, if $t\in\script{F}(k_1)$ then $t \notin \script{F}(k_2)$ as long as $k_1\neq k_2$. We can write an equation for the average delay as
\begin{equation}
\bW_i \triangleq \lim_{K \rightarrow \infty} \frac{\EE{\sum_{k=0}^K \lb\sum_{j\in \script{A}_i(k)}W_i^{(j)}\rb}}{\EE{\sum_{k=0}^{K}{\vert\script{A}_i(k)\vert}}}
\label{Delay_Frame}
\end{equation}
where $\script{A}_i(k)\triangleq\cup_{t\in\script{F}(k)}\script{A}_i^{(t)}$ is the set of all packets that arrive at $\SU{i}$'s buffer during frame $k$. We note that the long-term average delay $\bW_i$ in \eqref{Delay_Frame} depends on the chosen priority lists as well as the power allocation policy, in all frames $k\geq 0$.

\subsection{Problem Statement}
\label{Prob_Statement_Subsection}
We are interested to find the optimum scheduling-and-power-allocation policy that minimizes the sum of SUs' average delays subject to per-SU delay constraint as well as some interference constraints. In this paper, we consider two kinds of interference constraints: 1) instantaneous interference constraint; 2) average interference constraint. Since time-slot-based policies that update the scheduling and power-allocation each time-slot suffer from curse of dimensionality \cite{li2011delay}, we restrict our problem to frame-based scheduling policies as well as frame-based power allocation policies. The former is represented by the priority list $\bfpi(k)$ discussed earlier. On the other hand, the latter is defined in the following definition.

\begin{Def}
\label{Frame_Based_Pow_All}
A power allocation policy is said to be a \emph{frame-based power allocation policy} if, at each time slot $t\in\script{F}(k)$ the scheduled user transmits with power $\Pit$ on the form
\begin{equation}
	\Pit=\min\lb\frac{\Iinst}{g_i^{(t)}},P_i(k)\rb,
	\label{Pow_Allocation}
\end{equation}
where $P_i(k)$ is some constant that is fixed $\forall t\in\script{F}(k)$. We refer to $P_i(k)$ as the power parameter of $\SU{i}$.
\end{Def}
\noindent In future sections, we will show that restricting the power allocation policy to the frame-based power allocation policy does not result in loss of optimality.

Consider the following constraints
\begin{align}
& \bW_i \leq d_i & , & \forall i \in \script{N}
\label{Avg_Delay_C}\\
& \Pmin\leq P_i^{(t)} \leq \Pmax & , & \forall i \in \script{N} \rm{\; and \;}\forall t\geq 1,
\label{Max_Pow_C}\\
& \sum_{i=1}^N {P_i^{(t)}g_i^{(t)}} \leq \Iinst &, &\forall t\geq 1,
\label{Inst_Interf_C}\\
& \sum_{i=1}^N{\mathds{1} \lb P_i^{(t)}\rb} \leq 1 & , & \forall t \geq 1,
\label{Single_Tx_C}\\
& I\triangleq \lim_{T\rightarrow \infty} \sum_{i=1}^N\frac{1}{T}\sum_{t=1}^T P_i^{(t)}g_i^{(t)}\leq \Iavg,&&
\label{Avg_Interf_C}
\end{align}
%$h_i\lb \cdot \rb$ are some convex increasing functions called the ``cost functions'', 
where $I$ denotes the long-term average interference received by the PU while $\mathds{1}(x)\triangleq 1$ if $x\neq 0$ and $0$ otherwise. Constraint \eqref{Avg_Delay_C} is the average delay constraint for $\SU{i}$, \eqref{Max_Pow_C} is the maximum power constraint due to the limitations of $\SU{i}$'s transmitter as well as the minimum power constraint that results in finite delays for all SUs ($\Pmin$ is some constant that will be defined later), \eqref{Inst_Interf_C} is the instantaneous interference constraint for the PU, \eqref{Single_Tx_C} indicates that no more than a single SU is to be transmitting at slot $t$, while the last constraint \eqref{Avg_Interf_C} is to protect the PU from average interference. The two optimization problems that we solve in this paper are
\begin{align}
\underset{\{\bfpi(k)\},\{\bfP{}(k)\}}{\rm{minimize}}& \sum_{i=1}^N \bW_i &
\label{Problem}\\
\nonumber\rm{subject \; to\; } & \text{constraints \eqref{Avg_Delay_C}, \eqref{Max_Pow_C}, \eqref{Inst_Interf_C} and \eqref{Single_Tx_C}} &
\end{align}
and
\begin{align}
\underset{\{\bfpi(k)\},\{\bfP{}(k)\}}{\rm{minimize}}& \sum_{i=1}^N \bW_i
\label{Prob}\\
\nonumber\rm{subject \; to\; } & \text{constraints \eqref{Avg_Delay_C}, \eqref{Max_Pow_C}, \eqref{Inst_Interf_C}, \eqref{Single_Tx_C} and \eqref{Avg_Interf_C}}.
\end{align}
We refer to problem {\ProbA} as the \emph{instantaneous interference constraint problem}, while to {\ProbB} as the \emph{average interference constraint problem}. In the next section we solve these two problems and show that their solutions are different.

\section{Proposed Power Allocation and Scheduling Algorithm}
\label{Proposed_Algorithm}
We solve problems {\ProbA} and {\ProbB} by proposing online joint scheduling and power allocation policies that dynamically update the scheduling and the power allocation. We show that these policies have performances that come arbitrarily close to being optimal. That is, we can achieve a sum of the average delays arbitrarily close to its optimal value depending on some control parameter $V$.

We first discuss the idea behind our policies. Then we present the proposed policy for each problem, {\ProbA} and {\ProbB}, separately.

\subsection{Satisfying Delay Constraints}
%As will be discussed in Section \ref{Algorithm_DOIC}, the proposed algorithm is executed at the beginning of each frame. 
In order to guarantee a feasible solution satisfying the delay constraints in problems {\ProbA} and {\ProbB}, we set up a ``virtual queue'' associated with each delay constraint $\bW_i\leq d_i$. The virtual queue will be used in both problems {\ProbA} and {\ProbB}. The virtual queue for $\SU{i}$ at frame $k$ is given by
\begin{equation}
Y_i(k+1)\triangleq\lb Y_i(k)+\sum_{j\in \script{A}_i(k)}{\lb W_i^{(j)}-r_i(k)\rb} \rb^+
\label{Delay_Q}
\end{equation}
where $r_i(k)\in[0,d_i]$ is an auxiliary random variable, that is to be optimized over and $Y_i(0)\triangleq 0$, $\forall i$. We define $\bfY(k) \parFdef{Y}$. Equation \eqref{Delay_Q} is calculated at the end of frame $k-1$ and represents the amount of delay exceeding the delay bound $d_i$ for $\SU{i}$ up to the beginning of frame $k$. We use the definition of \emph{mean rate stability} as in \cite{li2011delay} to state the following lemma.%We first give the following definition, then state a lemma that gives a sufficient condition on $Y_i(k)$ for the delay of $\SU{i}$ to satisfy $\bW_i \leq d_i$.
%\begin{Def}
%\label{Mean_Rate_Def}
%A random sequence $\{Y_i(k)\}_{k=0}^\infty$ is mean rate stable if and only if the equality $\lim_{K\rightarrow\infty}\EE{Y_i(K)}/K=0$ holds.
%\end{Def}

\begin{lma}
\label{Mean_Rate_Lemma}
If $\{Y_i(k)\}_{k=0}^\infty$ is mean rate stable, then the time-average delay of $\SU{i}$ satisfies $\bW_i \leq d_i$.
\end{lma}
\begin{proof}
Following similar steps as in Lemma 3 in \cite{li2011delay}, we can show that
%Removing the $(\cdot)^+$ sign from \eqref{Delay_Q} yields
%\begin{equation}
%Y_i(k+1) \geq Y_i(k)+\sum_{j\in \script{A}_i(k)}{\lb W_i^{(j)}-r_i(k)\rb}.
%\label{Inequality1}
%\end{equation}
%Summing inequality \eqref{Inequality1} over $k=0,\cdots K-1$ and noting that $Y_i(0)=0$ gives
%\begin{equation}
%Y_i(K)\geq \sum_{k=0}^{K-1} \lb\sum_{j\in \script{A}_i(k)}W_i^{(j)}\rb-\sum_{k=0}^{K-1} {\lb r_i(k) \vert\script{A}_i(k)\vert\rb}.
%\label{Inequality2}
%\end{equation}
%Taking the $\EE{\cdot}$ then dividing by $\EE{\sum_{k=0}^{K-1}{\vert\script{A}_i(k)\vert}}$ gives
\begin{multline}
\frac{\EE{\sum_{k=0}^{K-1} \lb\sum_{j\in \script{A}_i(k)}W_i^{(j)}\rb}}{\EE{\sum_{k=0}^{K-1}{\vert\script{A}_i(k)\vert}}} \leq \\
\frac{\EE{Y_i(K)}}{K}\frac{K}{\EE{\sum_{k=0}^{K-1}{\vert\script{A}_i(k)\vert}}} + \frac{\sum_{k=0}^{K-1}{\EE{\vert\script{A}_i(k)\vert {r_i(k)}}}}{\sum_{k=0}^{K-1}\EE{\vert\script{A}_i(k)\vert}}.
\label{Wait_r_i}
\end{multline}
Replacing $r_i(k)$ by its upper bound $d_i$, taking the limit as $K\rightarrow\infty$ then using the mean rate stability definition and \eqref{Delay_Frame} completes the proof.
\end{proof}
Lemma \ref{Mean_Rate_Lemma} provides a condition on the virtual queue $\Yivq$ so that $\SU{i}$'s average delay constraint $\bW_i\leq d_i$ in \eqref{Avg_Delay_C} is satisfied. That is, if the proposed joint power allocation and scheduling policy results in a mean rate stable $\Yivq$, then $\bW_i\leq d_i$. For both problems, the proposed policy depends on the Lyapunov optimization where the goal is to choose the joint scheduling and power allocation policy that minimizes the drift-plus-penalty. In Section \ref{Algorithm_DOIC} (Section \ref{Algorithm_DOAC}) we will show that if problem {\ProbA} (problem {\ProbB}) is feasible, then the proposed policy guarantees mean rate stability for the queues $\Yivq$.

\subsection{Algorithm for Instantaneous Interference Constraint Problem}
\label{Algorithm_DOIC}
We now propose the \emph{Delay Optimal with Instantaneous Interference Constraint} (\emph{DOIC}) policy that solves problem {\ProbA}. This policy is executed at the beginning of each frame $k$ for finding $\bfP^{(t)}$ as well as the optimum list $\bfpi(k)$, given some prespecified control parameter $V$. Define the random variable $R_i(P)$ as (not to be confused with $R_i^{(t)}$ in \eqref{Tx_Rate})
\begin{equation}
R_i(P)\triangleq \Ts\log\lb 1+\min\lb \frac{\Iinst}{g_i^{(t)}},P\rb\gamma_i^{(t)}\rb,
\label{Rate_Explicit}
\end{equation}
where $P$ is some fixed constant argument and define $\mu_i(P)\triangleq\EE{R_i(P)}/L$ where the expectation is taken over $\git$ and $\gamma_i^{(t)}$. We now present the {\DOIC} policy, its optimality and then the intuition behind it.\\
\noindent{\bf {\DOIC} Policy} (executed at the beginning of frame $k$):
\begin{enumerate}
	\item The BS sorts the SUs according to the descending order of $Y_i(k)\mu_i(\Pmax)$. The sorted list is denoted by $\bfpi(k)$.
	\item At the beginning of each slot $t\in\script{F}(k)$ the BS schedules $\SU{i^*}$ that has the highest priority in the list $\bfpi(k)$ among those having non-empty buffers.
	\item $\SU{i^*}$, in turn, transmits $M_{i^*}^{(t)}$ packets as dictated by \eqref{Num_Bits} where $P_i^{(t)}=0$ $\forall i\neq i^*$ while $P_{i^*}^{(t)}$ is calculated as
\begin{equation}
	P_{i^*}^{(t)}=\min\lb\frac{\Iinst}{g_{i^*}^{(t)}},\Pmax\rb,
\label{Power_Allocation}
\end{equation}
	\item At the end of frame $k$, for all $i\in \script{N}$ the BS updates:
	\begin{enumerate}
		\item $r_i(k)= d_i$ if $V<Y_i(k)\lambda_i$, and $r_i(k)=0$ otherwise, and then
		\item $Y_i(k+1)$ via \eqref{Delay_Q}.
	\end{enumerate}
\end{enumerate}
Before we discuss the optimality of the {\DOIC} in Theorem \ref{Optimality}, we define the following quantities. Let $a\triangleq 1-\Pi_{i=1}^N \lb 1-\lambda_i \rb$ denote the probability of receiving a packet from a user or more at a given time slot, while $C_Y\triangleq\sum_{i=1}^N C_{Y_i}$ with $C_{Y_i}\triangleq \sqrt{\EE{A^4}\EE{B^4}} + d_i^2\EE{A^2}$, where $\EE{A^2}$ and $\EE{A^4}$ are bounds on the second and fourth moments of the total number of arrivals $\sum_i\vert \script{A}_i(k)\vert$ during frame $k$, respectively, while $\EE{B^4}$ is a bound on the fourth moment of the busy period $B(k)$. The finiteness of these moments can be shown to hold if the first four moments of the service time are finite. In Appendix \ref{No_Deep_Fade} we show that all the service time moments exist given any distribution for $\Pit\gamma_i^{(t)}$. We omit the derivation of these bounds due to lack of space.

\begin{thm}
\label{Optimality}
If problem {\ProbA} is strictly feasible, then the proposed {\DOIC} policy results in a time average of the SUs' delays satisfying the following inequality
\begin{equation}
\sum_{i=1}^N{\bW_i} \leq \frac{aC_Y}{V} + \sum_{i=1}^N{\bW_i^*},
\label{Optimality_Equation}
\end{equation}
where $\bW_i^*$ is the optimum value of the delay when solving problem {\ProbA}, while $a$ and $C_Y$ are as given above. Moreover, the virtual queues $\Yivq$ are mean rate stable $\forall i \in \script{N}$.
\end{thm}

\begin{proof}
See Appendix \ref{Optimality_Proof_Inst}.%The proof follows by setting $\Iavg=\infty$ ($\forall i\in\script{N}$) in the proof in Appendix \ref{Optimality_Proof}. Setting $\Iavg=\infty$, sets $X(k)=0$, $\forall k\geq 0$, by \eqref{Avg_Interf_Q}.
\end{proof}
%\begin{proof}
%The proof follows by setting $\Iavg=\infty$ as well as setting $\bP_i=\infty$ ($\forall i\in\script{N}$) in the proof in Appendix \ref{Optimality_Proof}. Setting $\Iavg=\infty$, sets $X(k)=0$, $\forall k=0,2,...$, by \eqref{Avg_Interf_Q}, while setting $\bP_i=\infty$ sets $Z_i(k)=0$, $\forall k=0,2,...$, by \eqref{Avg_Pow_Q}.
%\end{proof}

Theorem \ref{Optimality} says that the objective function of problem {\ProbA} is upper bounded by the optimum value $\sum_i\bW_i^*$ plus some constant gap that vanishes as $V\rightarrow\infty$. Having a vanishing gap means that the {\DOIC} policy is asymptotically optimal. Moreover, based on the mean rate stability of the queues $\Yivq$, the set of delay constraints of problem {\ProbA} is satisfied.

The intuition behind the {\DOIC} policy comes from the proof of Theorem \ref{Optimality}. In the proof, we follow the Lyapunov optimization technique to obtain an expression for the drift-plus-penalty then upper bound this expression (see \eqref{Drift_Plus_Penalty1}). The {\DOIC} policy becomes the one that minimizes this upper bound or, simply, minimizing $\Phi$ which is given by
\begin{equation}
\Phi_{\rm I}\triangleq \sum_{i=1}^N \lb V -Y_i(k) \lambda_i\rb r_i(k)+ \sum_{j=1}^N Y_{\pi_j}(k) \lambda_{\pi_j} \EEY{W_{\pi_j}^{(j)}}.
\label{Upper_Bound_Inst}
\end{equation}
Minimizing the first summation in $\Phi_{\rm I}$ minimizes objective function in {\ProbA}, while minimizing the second summation guarantees that the solution is feasible. We observe that the first term in \eqref{Upper_Bound_Inst} can be minimized independent of the second term. Step 4.a in the {\DOIC} policy minimizes the first term in \eqref{Upper_Bound_Inst} while, using the $c\mu$ rule \cite{c_mu_Rule}, the second term is minimized in Step 1.

In the {\DOIC} policy, the drawback of setting $V$ very large is that the time needed for the algorithm to converge increases. This increase is linear in $V$ \cite{NeelyPhD}. That is, if the number of frames required for the quantity $\sum_i{Y_i(k)}/(Nk)$ to be less than $\epsilon$ (for some $\epsilon>0$) is $O(K_1)$, then increasing $V$ to $\beta V$ will require $O(\beta K_1)$ frames for it to be less than $\epsilon$, for any $\beta>1$. We note that the complexity of the {\DOIC} policy is $O(N)$ because calculating $\mu_i(\Pmax)$ is of $O(1)$, while the power is closed-form in \eqref{Power_Allocation}. We note that if problem {\ProbA} is not feasible, then this is because one of two reasons; either one or more of the constraints is stringent, or otherwise because $\sum_{i=1}^N\lambda_i/\mu_i(\Pmax)\geq1$. If it is the former, then the {\DOIC} policy will result in a point that is as close as possible to the feasible region. On the other hand, if it is the latter, then we could add an admission controller that limits the average number of packets arriving at buffer $i$ to $\lambda_i(1-\epsilon)/\lb \sum_{i=1}^N\lambda_i/\mu_i(\Pmax)\rb$ for some $\epsilon>0$.

\subsection{Algorithm for Average Interference Constraint Problem}
\label{Algorithm_DOAC}
We now propose the Delay-Optimal-with-Average-Interference-Constraint {\DOAC} policy for problem {\ProbB}. We first give the following useful definitions. Since the scheduling scheme in frame $k$ is a priority scheduling scheme with preemptive-resume queuing discipline, then given the priority list $\bfpi$ we can write the expected waiting time of all SUs in terms of the average residual time \cite[pp. 206]{Bertsekas_Data_Networks} defined as $T_{\pi_j}^{\rm R}\triangleq \sum_{l=1}^j\lambda_{\pi_l}\EE{s_{\pi_l}^2}/2$, where the expectation is taken over $P_{\pi_l}^{(t)}\gamma_{\pi_l}^{(t)}$. The waiting time of SU $\pi_j$ that is given the $j$th priority is \cite[pp. 206]{Bertsekas_Data_Networks}
\begin{multline}
W_{\pi_j}\lb P,\mu_{\pi_j}(P),\rho_{\pi_j}(P),\brho_{\pi_{j-1}},T_{\pi_j}^{\rm R}\rb \triangleq\\
\frac{1}{\lb 1-\brho_{\pi_{j-1}}\rb}\left[\frac{1}{\mu_{\pi_j}(P)} + \frac{T_{\pi_j}^{\rm R}}{\lb 1-\brho_{\pi_{j-1}} - \rho_{\pi_j}(P)\rb}\right]
\label{Priority_Delay}
\end{multline}
where $\rho_i(P)\triangleq \lambda_i/\mu_i(P)$ and $\brho_{\pi_{j-1}} \triangleq \sum_{l=1}^{j-1} \rho_{\pi_l}(P_{\pi_l})$. Moreover, we define
\begin{multline}
\Wup\lb P,\rho_{\pi_j}(P),\brhomax_{\pi_{j-1}},T_{\pi_j}^{\rm R}\rb \triangleq\\
\frac{1}{\lb 1-\brhomax_{\pi_{j-1}}\rb}\left[\frac{1}{\mu_{\pi_j}(P)} + \frac{T_{\pi_j}^{\rm R}}{\lb 1-\brhomax_{\pi_{j-1}} - \rho_{\pi_j}(P)\rb}\right]
\label{Priority_Delay_UB}
\end{multline}
where %$T^{\rm R}$ is an upper bound on $T_{\pi_j}^{\rm R}$ and is given by $T^{\rm R}\triangleq\sum_{i=1}^N\lambda_i\lb L^2+L\lb 1-p_i(\Pminn_i)\rb\rb/p_i^2(\Pminn_i)/2$ with $p_i(P)\triangleq1-\Prob{R_i(P)=0}$ and $\Pminn_i$ is the minimum power satisfying $\rho_i(\Pminn_i)+\sum_{j\neq i}\rho_j(\Pmax)<1$ (see Appendix \ref{No_Deep_Fade} for the derivation of $T^{\rm R}$), while 
$\brhomax_i$ is some upper bound on $\brho_i$ that will be defined later. We henceforth drop all the arguments of $\Wup(P,\brhomax_{\pi_{j-1}})$ except $P$ and $\brhomax_{\pi_{j-1}}$ and all those of $W_{\pi_j}(P)$ except $P$.

To track the average interference at the PU up to the end of frame $k$ we set up the following virtual queue that is associated with the average interference constraint in problem {\ProbB} and is calculated at the BS at the end of frame $k$.
\begin{equation}
X(k+1)\triangleq \lb X(k) + \sum_{i=1}^N{\sum_{t\in\script{F}(k)}{P_i^{(t)} g_i^{(t)}}} -\Iavg T_k\rb^+,
\label{Avg_Interf_Q}
\end{equation}
where the term $\sum_{i=1}^N{\sum_{t\in\script{F}(k)}{P_i^{(t)}g_i^{(t)}}}$ represents the aggregate amount of interference energy received by the PU due to the transmission of the SUs during frame $k$. Hence, this virtual queue is a measure of how much the SUs have exceeded the interference constraint above the level $\Iavg$ that the PU can tolerate. Lemma \ref{Mean_Rate_Lemma_Avg_Interf} provides a sufficient condition for the interference constraint of problem {\ProbB} to be satisfied.

\begin{lma}
\label{Mean_Rate_Lemma_Avg_Interf}
If $\{X(k)\}_{k=0}^\infty$ is mean rate stable, then the time-average interference received by the PU satisfies $I \leq \Iavg$.
\end{lma}
\begin{proof}
The proof is similar to that of Lemma \ref{Mean_Rate_Lemma} and is omitted for brevity.
\end{proof}
Lemma \ref{Mean_Rate_Lemma_Avg_Interf} says that if the power allocation and scheduling algorithm results in mean rate stable $\Xvq$, then the interference constraint of problem {\ProbB} is satisfied.

Before presenting the {\DOAC} policy, we first discuss the idea behind it. Intuitively, a policy that solves problem {\ProbB} should allocate $\SU{i}$'s power and assign its priority such that $\SU{i}$'s expected delay and the expected interference to the PU is minimized. The {\DOAC} policy is defined as the policy that selects the power parameter vector $\bfP(k)\parFdef{P}$ jointly with the priority list $\bfpi(k)$ that minimizes $\Psi\triangleq\sum_{j=1}^N\psi_{\pi_j}(P_{\pi_j}(k),\brhomax_{\pi_{j-1}})$ where
\begin{equation}
\label{Optimization_Obj}\psi_{\pi_j}(P,\brhomax_{\pi_{j-1}})\triangleq \psiDel_{\pi_j}(P,\brhomax_{\pi_{j-1}}) + \psiInt_{\pi_j}(P), %\hspace{0.1in} {\rm with}\\
%\psiDel_{\pi_j}(P,\brhomax_{\pi_{j-1}})&\triangleq Y_{\pi_j}(k) \lambda_{\pi_j} \Wup(P,\brhomax_{\pi_{j-1}}), \hspace{0.1in} {\rm while}\\
%\psiInt_{\pi_j}(P)&\triangleq X(k)\rho_{\pi_j}(P)P\bar{g}_{\pi_j}.
\end{equation}
with $\psiDel_{\pi_j}(P,\brhomax_{\pi_{j-1}})\triangleq Y_{\pi_j}(k) \lambda_{\pi_j} \Wup(P,\brhomax_{\pi_{j-1}})$ while $\psiInt_{\pi_j}(P)\triangleq X(k)\rho_{\pi_j}(P)P\bar{g}_{\pi_j}$. The function $\psiDel_{\pi_j}(P,\brhomax_{\pi_{j-1}})$ (and $\psiInt_{\pi_j}(P)$) represents the amount of delay (interference) that SU $\pi_j$ is expected to experience (to cause to the PU) during frame $k$.

The brute search of $\bfP{}(k)$ and $\bfpi(k)$ that minimizes $\Psi$ is exponentially high. To minimize $\Psi$ in a computationally efficient way, we need the functions $\psi_{\pi_j}(P_{\pi_j}(k),\brhomax_{\pi_{j-1}})$ to become decoupled for all $j\in\script{N}$. That is, we want $\psi_{\pi_j}(P_{\pi_j}(k),\brhomax_{\pi_{j-1}})$ not to depend on $P_{\pi_l}(k)$ as long as $l\neq j$. Hence, we set the function $\brhomax_{\pi_{j-1}}$ to some function that does not depend on the optimization power variables $P_{\pi_l}(k)$ for all $l\leq j-1$ but otherwise on some other fixed parameters. We need to choose these parameters such that the bound
\begin{equation}
\brhomax_{\pi_{j-1}}\geq\brho_{\pi_{j-1}} \triangleq \sum_{l=1}^{j-1} \rho_{\pi_l}(P_{\pi_l})
\label{brhomax_Bound}
\end{equation} is satisfied. Thus, these functions, are given by
\begin{equation}
\brhomax_{\pi_{j-1}}\triangleq \sum_{l=1}^{j-1} \rho_{\pi_l}\lb P_{\pi_l}^{\brhomax}\rb,
\label{brhomax}
\end{equation}
where
\begin{equation}
P_{\pi_l}^{\brhomax}\triangleq\arg \min_P \psi_{\pi_l}\lb P,\brhomax_{\pi_{l-1}}\rb.
\label{Local_psi_Inequality}
\end{equation}
With $\brhomax_{\pi_{j-1}}$ given by \eqref{brhomax}, $\psi_{\pi_j}(P_{\pi_j}(k),\brhomax_{\pi_{j-1}})$ is a function in $P_{\pi_j}(k)$ only. Before we show that the choice of \eqref{brhomax} and \eqref{Local_psi_Inequality} guarantees that \eqref{brhomax_Bound} is satisfied, we note that \eqref{brhomax} dictates that in order to find $\brhomax_{\pi_{j-1}}$ we need to find $P_{\pi_l}^{\brhomax}$ for all $l<j-1$. Hence, we find $P_{\pi_j}^{\brhomax}$ recursively starting from $j=1$ at which $\brhomax_{\pi_0}=0$ by definition. It is shown in \cite[Lemma 5, pp. 55]{Ewaisha_Comprehensive} that $\brhomax_{\pi_j}$ is an upper bound on $\brho_{\pi_j}$. $\brhomax_{\pi_j}$ has an advantage over $\brho_{\pi_j}$ (and hence $\psi_{\pi_j}\lb P_{\pi_j},\brhomax_{\pi_{j-1}}\rb$ over $\psi_{\pi_j}\lb P_{\pi_j},\brho_{\pi_{j-1}}\rb$) which is that it is not a function in $P_{\pi_l}$ for $l\neq j$. This decouples the power search optimization problem to $N$ one-dimensional searches.

\begin{algorithm}
\caption{{\DOACopt}: Optimization-problem-solution algorithm called by the {\DOAC} policy at the beginning of frame $k$ to solve for $\bfPst(k)$ as well as $\bfpist(k)$.}
\begin{algorithmic}[1]
\label{DOACopt}
\STATE Define $\script{S}$ as the set of all sets formed of all subsets of $\script{N}$ and define the auxiliary functions% $\tPsi(\cdot,\cdot):\script{N}\times\sS\rightarrow \mathbb{R}^+$, $\trho(\cdot):\sS\rightarrow [0,1]$, $\tS(\script{X}):\sS\rightarrow\script{N}^{\vert\script{X}\vert}$, $\tP(\script{X}):\sS\rightarrow[0,\Pmax]^{\vert\script{X}\vert}$, $\bP(\cdot,\cdot):\sS\times\script{N}\rightarrow[0,\Pmax]$.
\begin{align*}
&\tPsi(\cdot,\cdot):\script{N}\times\sS\rightarrow \mathbb{R}^+\\
& \trho(\cdot):\sS\rightarrow [0,1],\\
&\tS(\script{X}):\sS\rightarrow\script{N}^{\vert\script{X}\vert},\\
&\tP(\script{X}):\sS\rightarrow[0,\Pmax]^{\vert\script{X}\vert},\\
&\bP(\cdot,\cdot):\sS\times\script{N}\rightarrow[0,\Pmax].
\end{align*}
\STATE Initialize $\tPsi(0,\cdot)=0$, $\trho(\phi)=0$, $\tS(\phi)=[\hspace{0.05in}]$ and $\tP(\phi)=[\hspace{0.05in}]$, where $\phi$ is the empty set.
\FOR{$i=1,\cdots, N$}
\STATE In stage $i$, the first $i$ priorities have been assigned to $i$ users. The corresponding priority list is denoted $[\pi_1,\cdots,\pi_i]$. In stage $i$ we have $\binom{N}{i}$ states each corresponds to a set $j$ formed from all possible combinations of $i$ elements chosen from the set $\script{N}$. We calculate $\tPsi(i,j)$ associated with each state $j$ in terms of $\tPsi(i-1,\cdot)$ obtained in stage $i-1$ as follows.
\FOR{$j\in$ all possible $i$-element sets}
\STATE At state $j\triangleq \{\pi_1,\cdots,\pi_i\}$, we have $i$ transitions, each connects it to state $j'$ in stage $i-1$, where $j'\triangleq j\backslash l$ with $l\in j$. Find the power associated with each transition $l\in j$ denoted $\bP(j,l)\triangleq\arg\min_P \psi_l(P,\trho(j\backslash l))$.
\STATE Set%$l^*=\arg\min_{l\in j}\tPsi\lb i-1,j\backslash l\rb+ \psi_l\lb\bP(j,l),\trho(j\backslash l)\rb$, $\tPsi(i,j)=\tPsi(i-1,j\backslash l^*)+ \psi_{l^*}\lb\bP(j,l^*),\trho(j\backslash l^*)\rb$, $\trho(j)=\trho\lb j\backslash l^*\rb+\rho\lb\bP(j,l^*)\rb$, $\tS(j)=\left [\tS\lb j\backslash l^*\rb,l^*\right]^T$ and $\tP(j)=\left [\tP\lb j\backslash l^*\rb,\bP(j,l^*)\right]^T$.
\begin{align*}
%\label{Survivor_Path}
&l^*=\arg\min_{l\in j}\tPsi\lb i-1,j\backslash l\rb+ \psi_l\lb\bP(j,l),\trho(j\backslash l)\rb,\\
%\label{Survivor_Equ}
&\tPsi(i,j)=\tPsi(i-1,j\backslash l^*)+ \psi_{l^*}\lb\bP(j,l^*),\trho(j\backslash l^*)\rb,\\
%\label{Survivor_rho}
&\trho(j)=\trho\lb j\backslash l^*\rb+\rho\lb\bP(j,l^*)\rb,\\
%\label{Survivor_User}
&\tS(j)=\left [\tS\lb j\backslash l^*\rb,l^*\right]^T,\\
%\label{Survivor_Pow}
&\tP(j)=\left [\tP\lb j\backslash l^*\rb,\bP(j,l^*)\right]^T.
\end{align*}
\ENDFOR
\ENDFOR
\STATE Set $\bfpist(k)= \tS\lb\script{N}\rb$ and $\bfPst(k)=\tP\lb\script{N}\rb$.
\end{algorithmic}
\end{algorithm}
%\STATE Initialize $P_{\pi_j}^{1} \leftarrow \Pmin$, $P_{\pi_j}^{2} \leftarrow \Pmin+\phi\lb \Pmax - \Pmin \rb$, $P_{\pi_j}^{3} \leftarrow \Pmax$, $P_{\pi_j}^{4} \leftarrow P_{\pi_j}^{2}+\phi\lb P_{\pi_j}^{3} - P_{\pi_j}^{2} \rb$, $\psi_{\pi_j}^{(l)}=\psi_{\pi_j}(P_{\pi_j}^{l})$, $\forall l=1,2,3,4$%\psi_{\pi_j}^{(2)}=\psi_{\pi_j}(P_{\pi_j}^{2})$ and $\psi_{\pi_j}^{(3)}=\psi_{\pi_j}(P_{\pi_j}^{3})$
%\WHILE{$\vert P_{\pi_j}^{3}-P_{\pi_j}^{1}\vert \leq \epsilon \lb P_{\pi_j}^{2}+P_{\pi_j}^{4}\rb$}
%\STATE Set $\psi_{\pi_j}^{(4)}=\psi_{\pi_j}(P_{\pi_j}^{4})$.
%\IF{$\psi_{\pi_j}^{(4)}>\psi_{\pi_j}^{(2)}$}
%\STATE Set $P_{\pi_j}^{3} \leftarrow P_{\pi_j}^{4}$ and $P_{\pi_j}^{4} \leftarrow P_{\pi_j}^{1}+\phi\lb P_{\pi_j}^{2} - P_{\pi_j}^{1} \rb$
%\ELSE
%\STATE Set $P_{\pi_j}^{1} \leftarrow P_{\pi_j}^{2}$, $P_{\pi_j}^{2} \leftarrow P_{\pi_j}^{3}$ and $P_{\pi_j}^{4} \leftarrow P_{\pi_j}^{2}+\phi\lb P_{\pi_j}^{3} - P_{\pi_j}^{2} \rb$
%\ENDIF
%\ENDWHILE
After reducing the search complexity of the power vector, we reduce the search complexity of the priority list from $N!$ to $2^N$. To do this, we use the dynamic programming illustrated in Algorithm \ref{DOACopt} that solves $\min_{\bfpi(k),\bfP{}(k)} \Psi$. Its search complexity is of $O(MN2^N)$ where $M$ is the number of iterations in a one-dimensional search, while $O(1)$ is the complexity of calculating $\Psi$ for a given priority list $\bfpi(k)$ and a given power vector $\bfP{}(k)$. Compared to the complexity of $O(M^N \cdot N!)$ which is that of the $N$-dimensional power search along with the brute-force of all $N!$ permutations of priority list $\bfpi(k)$, this is a large complexity reduction. However, the $O(MN2^N)$ is still high if $N$ was large. Finding an optimal algorithm with a lower complexity is extremely difficult since the scheduling and power control problem are coupled. In other words, in order to find the optimum scheduler we need to know the optimum power vector and vice versa. In Section \ref{Suboptimal} we propose a sub-optimal policy with a very low complexity and little degradation in the delay performance. We now present the {\DOAC} policy that the BS executes at the beginning of frame $k$.

{\bf {\DOAC} Policy} (executed at the beginning of frame $k$):
\begin{enumerate}
	\item The BS executes {\DOACopt} in Algorithm \ref{DOACopt} to find the optimum power parameter vector $\bfPst(k) \parFdef{P^*}$ as well as the optimum priority list $\bfpist(k) \parFdef{\pi^*}$ that will be used during frame $k$.
	\item The BS broadcasts the vector $\bfPst(k)$ to the SUs.
	\item At the beginning of each slot $t\in\script{F}(k)$, the BS schedules $\SU{i^{*(t)}}$ that has the highest priority in the list $\bfpist(k)$ among those having non-empty buffers.
	\item $\SU{i^{*(t)}}$, in turn, transmits $M_{i^{*(t)}}^{(t)}$ bits as dictated by \eqref{Num_Bits} where $P_i^{(t)}=0$ for all $i\neq i^{*(t)}$ while $P_{i^{*(t)}}^{(t)}$ is given by \eqref{Pow_Allocation}.
	%\begin{equation}
%\begin{array}{ll}
%\underset{\{\bfP(k)\},\{\bfpi(k)\}}{\rm{minimize}}& \sum_{i=1}^N \bW_i \\
%\label{Prob_P_k}
%\rm{subject \; to} & I \leq \Iavg\\
%& \bW_i \leq d_i\\
%& \Pmin \leq P_i(k) \leq \Pmax \hspace{0.1in} , \hspace{0.1in} \forall i \in \script{N} \rm{\; and \;}\forall k \geq 1,\\
%\end{array}
%\end{equation}
%where $\Pmin$ is the minimum power that can be allocated to $\SU{i}$ that makes $\sum_i\rho_i\leq 1$, $\forall i\in \script{N}$, and
	%\item At the end of frame $k$, the BS updates $Y_i(k)$ $\forall i\in \script{N}$ as well as $X(k)$ via equations \eqref{Delay_Q} and \eqref{Avg_Interf_Q}, respectively.
	\item At the end of frame $k$, for all $i\in \script{N}$ the BS updates:
	\begin{enumerate}
		\item $r_i(k)= d_i$ if $V<Y_i(k)\lambda_i$, and $r_i(k)=0$ otherwise.
		\item $X(k+1)$ via \eqref{Avg_Interf_Q}.
		\item $Y_i(k+1)$ via \eqref{Delay_Q}, $\forall i\in \script{N}$.
	\end{enumerate}
\end{enumerate}
Define $C_X\triangleq\lb(1-a)(2+a)+\EE{B^2}+2\EE{B}(a-a^2)\rb\times\lb\Pmax^2\gmax^2+\Iavg^2\rb/a^2$ and $C\triangleq C_Y+C_X$ where $\EE{B}$ is a bound on the mean of $B(k)$. It can be shown that $\EE{B}$ and $\EE{B^2}$ are finite since the first two moments of the service time are finite (see Appendix \ref{No_Deep_Fade}). Thus, $C_X$ is finite. Next, we state Theorem \ref{Optimality_Avg} that discusses the optimality of the {\DOAC} policy.

\begin{thm}
\label{Optimality_Avg}
If \eqref{Prob} is strictly feasible and the BS executes the {\DOAC} policy, the time average of the SUs' delays satisfy the following inequality in the light traffic regime
\begin{equation}
\sum_{i=1}^N{\bW_i} \leq \frac{aC}{V} + \sum_{i=1}^N{\bW_i^*},
\label{Optimality_Equation_Avg}
\end{equation}
where $\bW_i^*$ is the optimum value of the delay when solving problem {\ProbB}. Moreover, the virtual queues $\Xvq$ and $\Yivq$ are mean rate stable $\forall i \in \script{N}$.
\end{thm}

\begin{proof}
See Appendix \ref{Optimality_Proof}.
\end{proof}
%\begin{proof}
%The proof follows by setting $\bP_i=\infty$ and, hence, setting $\{Z_i(k)=0\}_{k=0}^\infty$ by \eqref{Avg_Pow_Q}, $\forall i\in\script{N}$, in the proof in Appendix \ref{Optimality_Proof}.
%\end{proof}

Similar to Theorem \ref{Optimality}, Theorem \ref{Optimality_Avg} says that the interference and delay constraints of problem {\ProbB} are satisfied since the virtual queues $\Xvq$ and $\Yivq$ are mean rate stable. Hence, the performance of the {\DOAC} policy is asymptotically optimal.

The intuition behind the {\DOAC} policy is similar to that behind the {\DOIC} policy with some differences stated here. When upper bounding the drift-plus-penalty term, we obtain the expression $\sum_{i=1}^N \lb V -Y_i(k) \lambda_i\rb r_i(k)+\Psi$ where $\Psi$ is defined before \eqref{Optimization_Obj}. Minimizing the first term in this bound is carried out in Step 5.a of the {\DOAC} policy. On the other hand, minimizing $\Psi$ is carried out using the dynamic programming in Algorithm \ref{DOACopt}. The dynamic programing finds the optimum values of the two vectors $\bfpi(k)$ and $\bfP(k)$ in an efficient way of complexity $O(NM2^N)$ without having to calculate the objective function $\Psi$ for the whole sample space of size $N!\times M^N$. The reason we were able to use this algorithm is because we were able to find an upper bound $\Wup$ that does not depend on the vector $\bfpi(k)$, a property that is necessary for the dynamic programming and that is absent in $\bW_{\pi_j}$.

%\section{Average Per SU Power Constraint}
%\label{Avg_Pow}
%\begin{equation}
%Z_i(k+1)\triangleq \lb Z_i(k) + \sum_{t\in\script{F}(k)}{P_i^{(t)}} -\bP_i \rb^+
%\label{Avg_Pow_Q}
%\end{equation}

\subsection{Near-Optimal Low Complexity Algorithm for Average Interference Constraint Problem}
\label{Suboptimal}
As seen in the {\DOAC} policy, the complexity of finding the optimal power vector and priority list can be high when the number of SUs $N$ is large. This is mainly due to the large complexity of Algorithm \ref{DOACopt}. In this subsection we propose a suboptimal solution with an extreme reduction in complexity and with little degradation in the performance. This solution solves for the power allocation and scheduling algorithm, thus it replaces the Algorithm \ref{DOACopt}.

%\subsubsection{Background and Challenges}
The challenges in Algorithm \ref{DOACopt} are three-fold. First finding the priority list (scheduling problem) requires the search over $N!$ possibilities. Second, even with a genie-aided knowledge of the optimum list, we still have to carry-out $N$ one-dimensional searches to find $\bfPst(k)$ (power control problem). Third, the scheduling and power control problems are coupled. We tackle the latter two challenges first, by finding a low-complexity power allocation policy that is independent of the scheduling algorithm. Then we use the $c\mu$ rule \cite{c_mu_Rule} to find the priority list. The $c\mu$ rule is a policy that gives the priority list that minimizes the quantity $\sum_{i=1}^N Y_i(k)\lambda_i W_i(P_i(k))$, given some power allocation vector $\bfP(k)$.

For each priority list $\bfpi$ Algorithm \ref{DOACopt} minimizes $\psi_{\pi_j}(P)\triangleq \psiD(P)+\psiI(P)$ for each $\SU{i}$. %We now state the following lemma to discuss the monotonicity of $\psiD$ and $\psiI$.
%\begin{lma}
%\label{Mono_Subterms}
%The functions $\psiD$ and $\psiI$ are, respectively, monotonically decreasing and increasing in $P_{\pi_j}(k)$.
%\end{lma}
%\begin{proof}
%See Appendix \ref{Mono_Subterms_Apx}.
%\end{proof}
%Based on the monotonicity of these two terms, the optimum solution $P_{\pi_j}^*(k)$ will be close to either $\Pmin$ or $\Pmax$. For example, if $X(k)\gg Y_{\pi_j}(k)$ then $P_{\pi_j}^*(k)\approx \Pmin$ since $\psiI$, which is increasing in $P_{\pi_j}(k)$, dominates over $\psiD$. On the other hand, if $X(k)\ll Y_{\pi_j}(k)$ then $P_{\pi_j}^*(k)\approx \Pmax$. 
Define $\Pmin$ to be the minimum power that satisfies $\sum_{j=1}^N\rho_{\pi_j}(\Pmin)<1$. Intuitively, if, for some $\pi_j\in\script{N}$, $X(k)\gg Y_{\pi_j}(k)$ then $P_{\pi_j}^*(k)$ is expected to be close to $\Pmin$ since the interference term $\psiI(P)$ dominates over $\psiD(P)$ in the $\pi_j$th term of the summation in \eqref{Optimization_Obj}. On the other hand, if $X(k)\ll Y_{\pi_j}(k)$ then $P_{\pi_j}^*(k)\approx \Pmax$. We propose the following power allocation policy for $\SU{\pi_j}$  $\forall\pi_j\in\script{N}$
\begin{equation}
\hat{P}_{\pi_j}(k)=
\left\{
\begin{array}{lll}
	\Pmin \mbox{ if } X(k)>Y_{\pi_j}(k)\\
	\Pmax \mbox{ otherwise.}
\end{array}
\right.
\label{Subopt_Power}
\end{equation}
We can see that the power allocation policy in \eqref{Subopt_Power} does not depend on the position of $\SU{i}$ in the priority list as opposed to Algorithm \ref{DOACopt} which requires the knowledge of $\SU{\pi_j}$'s priority position. In other words, $\hat{P}_{\pi_j}(k)$ is a function of $\pi_j$ but it is not a function of $j$. %but the knowledge of the SUs preceding $\SU{i}$ in the list as well.
Before proposing the scheduling policy, we note the following two properties based on the knowledge of the power $\bfPst(k)$. First, when $X(k)=0$, the solution to the minimization problem $\min_\bfpi \Obj$ is given by the $c\mu$ rule \cite{c_mu_Rule} that sorts the SUs according to the descending order of $Y_{\pi_j}(k)\mu_{\pi_j}(\hat{P}_{\pi_j}(k))$. Second, when $Y_{\pi_j}(k)=0$ $\forall \pi_j \in \script{N}$, any sorting order would not affect the objective function $\Obj$.

%We also note that we could not make any use of these two properties to reduce the complexity of Algorithm \ref{DOACopt} because they hold only if the problem were to be a scheduling problem. However for joint power-allocation-and-scheduling problems these two properties do not hold in general. % Hence the proposed scheduling policy is to sort the SUs in a descending order of $Y_{\pi_j}(k)\mu_{\pi_j}(P_{\pi_j}(k))$, where $\mu_{\pi_j}(P_{\pi_j}(k))$ is calculated by setting .

%\subsubsection{Proposed Algorithm and Complexity}
The two-step scheduling and power allocation algorithm that we propose is 1) allocate the power vector $\bfP(k)$ according to \eqref{Subopt_Power}, then 2) assign priorities to the SUs in a descending order of $Y_{\pi_j}(k)\mu_{\pi_j}(\hat{P}_{\pi_j}(k))$ (the $c\mu$ rule). The complexity of this algorithm is that of sorting $N$ numbers, namely $O(N\log(N))$. This is a very low complexity if compared to that of the {\DOAC} policy of $O(MN \cdot N!)$. In Section \ref{Results} we will demonstrate that this huge reduction of complexity causes little degradation to the delay performance.

\section{Simulation Results}
\label{Results}
We simulated a system of $N=5$ SUs. Unless otherwise specified, Table \ref{Parameters} lists all parameter values for both scenarios; the instantaneous as well as the average interference constraint. $\SU{i}$'s arrival rate is set to $\lambda_i=i\lambda$ for some fixed parameter $\lambda$. All SUs are having homogeneous channel conditions except $\SU{5}$ who has the highest average interference channel gain. Thus $\SU{5}$ is statistically the worst case user. We assume that the SUs' delay constraints are $d_i=\dHigh$ $\forall i\leq 4$, and $d_5=\dLow$. In practice, $\Ts$ is around 1ms. We have chosen the values of $d_i$ to provide stringent QoS guarantees based on the $150ms$ average delay value for video packets recommended by CISCO (see \cite{Cisco_QoS}).% Fig. \ref{PerUser_Delay_Avg_N5} plots the per-user  delay using the {\DOAC} policy for two cases; the first is with $d_5=\dLow$ while the second is with $d_5=\dHigh$, to show the effect of an active versus an inactive delay constraint. In the active constrained case, the {\DOAC} policy guarantees that $\bW_5\leq d_5$. We conclude that the delay constraints in problem \eqref{Prob} can force the SUs' delays to take any strictly feasible value.
\begin{table}
	\centering
		\caption{Simulation Parameter Values}
		\label{Parameters}
		\begin{tabular}{|c|c||c|c|}
			\cline{1-4}
			Parameter & Value & Parameter & Value\\
			\cline{1-4}
			$(d_1,\cdots d_4,d_5)$ & $(\dHigh,\cdots,\dHigh,\dLow)\Ts$ & $\bgamma_i$ & $1$\\
			$\gammamax$ & $10\bgamma_i$ & $\Iinst$ & 20 \\
			$\gmax$ & $10\bgi$ & $\Pmax$ & 100 \\
			$\fgammai(\gamma)$ & $\exp{\lb-\gamma/\bgamma_i\rb}/\bgamma_i$ & $\alpha$ & 0.1 \\
			$\fgi(g)$ & $\exp{\lb-g/\overline{g}_i\rb}/\overline{g}_i$ & $\epsilon$ & $0.1$ \\
			%$\fgi(g)$ & $\delta(g-\overline{g}_i)$\\
			 $L$ & $1000$ bits/packet & $V$ & $100$ \\
			$(\overline{g}_1,\cdots\overline{g}_4,\overline{g}_5)$ & $(0.1,\cdots,0.1,0.4)$ & $\Iavg$& 5
			\\ \cline{1-4}
			\end{tabular}
\end{table}

\subsection{Per-user Performance}
We first consider problem {\ProbB} since it is more general. Fig. \ref{PerUser_Delay_Avg_N5} plots average per-SU delay $\bW_i$, from \eqref{Delay_Frame}, versus $\lambda$ assuming perfect knowledge of the direct and interference channel state information (CSI), namely $\gamma_i^{(t)}$ and $g_i^{(t)}$. The plot is for the {\DOAC} policy for two cases; the first being the constrained case where $d_5=\dLow\Ts$, while the second is the unconstrained case where $d_5=\dHigh\Ts$. We call it the unconstrained problem because the average delay of all SUs is strictly below $\dLow\Ts$, thus all delay constraints are inactive. We choose to compare these two cases to show the effect of an active versus an inactive delay constraint. From Fig. \ref{PerUser_Delay_Avg_N5} we can see that $\SU{5}$ has the worst average delay. However, for the constrained case, the {\DOAC} policy has forced $\bW_5$ to be smaller than $\dLow\Ts$ for all $\lambda$ values. This comes at the cost of another user's delay. We conclude that the delay constraints in problem {\ProbA} can force the delay vector of the SUs to take any value as long as it is strictly feasible.

\subsection{Total System's Delay Performance}
In Fig. \ref{Sum_Delay_Avg_DOAC_CSMA_Subopt_CSI_Unconst}, we compare the aggregate delay performance of seven different schemes following the parameters in Table \ref{Parameters} unless otherwise specified; 1) Cognitive Network Control policy proposed in \cite{Neely_CNC_2009} which is a version of the MaxWeight scheduling; 2) Carrier-Sense-Multiple-Access (CSMA) that assigns the channel equally likely to all users while allocating the same power as the {\DOAC} policy (genie-aided power allocation), 3) {\DOAC} in the presence of channel state information (CSI) errors; 4) Suboptimal policy proposed in Section \ref{Suboptimal}, 5) The constrained {\DOAC} case (or simply the {\DOAC}), 6) The {\DOIC} policy that neglects the average interference constraint; and 7) The Unconstrained {\DOAC} case having $d_5=\dHigh\Ts$. In the presence of CSI errors, we assumed that each SU has an error of $\alpha=10\%$ in estimating each of $\gamma_i^{(t)}$ and $g_i^{(t)}$. The actual and observed values of $\gamma_i^{(t)}$ and $g_i^{(t)}$ are related by $\gamma_i^{(t)}=\frac{\gammaerr}{1+\alpha/2}$ and $g_i^{(t)}=\frac{\gerr}{1-\alpha/2}$, respectively. In order to avoid outage we substitute by $\gamma_i^{(t)}$ in \eqref{Tx_Rate} while to guarantee protection to the PU from interference, we substitute $g_i^{(t)}$ in \eqref{Pow_Allocation} for the {\DOAC} policy.

In Fig. \ref{Sum_Delay_Avg_DOAC_CSMA_Subopt_CSI_Unconst} the relative delay gap between the perfect and imperfect CSI is around $5\%$ and $9\%$ at light and high traffic, respectively. The performance of this error model represents an upper bound on the actual difference since $\alpha=10\%$ is usually an upper bound on the actual estimation error. When implementing the suboptimal algorithm we find that the sum delay across SUs is very close to its optimal value found via Algorithm \ref{DOACopt}. This holds for both light and heavy traffics with delay performance gaps $0.06\%$ and $0.3\%$, respectively and they both outperform the CSMA and the CNC. This is because the proposed policies prioritize the users based on their delay and interference realizations. On the other hand, the CSMA allocates the channel to guarantee fairness of allocation across time and the CNC's goal is to maximize the achievable rate region \cite{li2011delay}.
%Fig. \ref{Sum_Delay_Avg_DOAC_CSMA_Subopt_CSI_Unconst} plots the sum of cost functions versus $\lambda$ for the perfect CSI case for both the constrained and unconstrained problem. Comparing the two curves, we find that the constrained problem has worse sum of cost functions for $\lambda\geq 0.6$. This is because the constrained problem has a smaller feasible region than that of the unconstrained one. We note that the PU's interference constrained is satisfied with probability 1 based on \eqref{Pow_Allocation}.% The intuition behind this is as follows. The WQW's main goal is to minimize the sum delay, thus when the interference caused by $\SU{1}$ is higher, then we expect an increase in its delay compared to $\SU{2}$'s.
\begin{figure}%
\centering
\includegraphics[width=0.95\columnwidth]{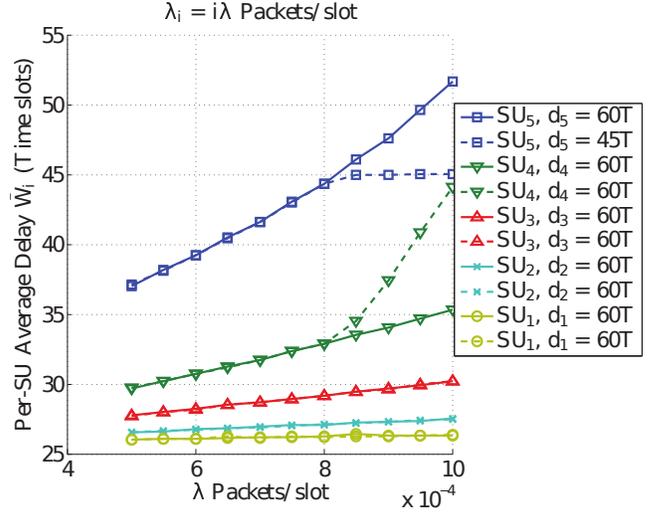}%
\caption{Average per-SU delay for both the constrained and unconstrained cases. Both cases are simulated using the {\DOAC} policy. $\SU{5}$ is the user with the worst channel statistics and the largest arrival rate. The {\DOAC} can guarantee a bound on $\bW_5$.}%
\label{PerUser_Delay_Avg_N5}%
\end{figure}

\begin{figure}%
\centering
\includegraphics[scale=0.5]{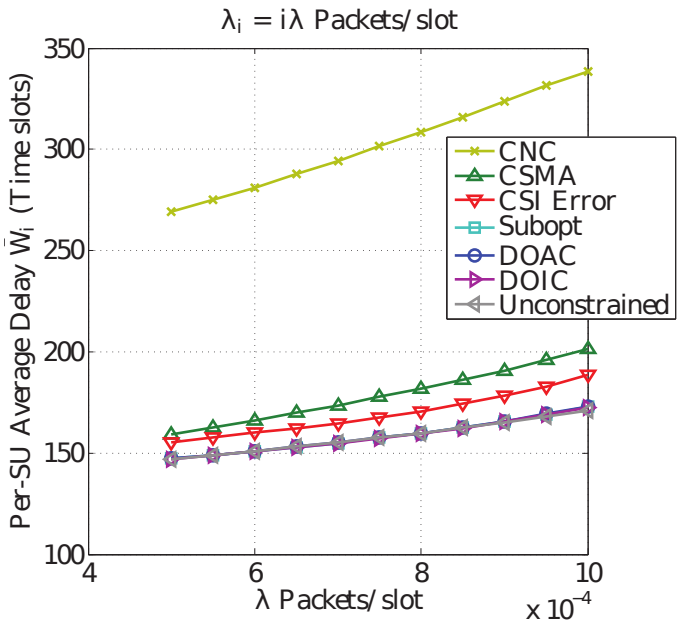}%
\caption{The average delay performance of seven schemes. The {\DOAC} and the suboptimal policies are within $0.3\%$, while both outperform the CSMA and the CNC by more than $8.2\%$ and $83\%$, respectively. The {\DOAC} under CSI errors experiences less than $9\%$ increase in the delay.% gives a very close delay performance to the DOAC, while the DOAC does not sacrifice the delay in the presence of channel state information errors.
}
\label{Sum_Delay_Avg_DOAC_CSMA_Subopt_CSI_Unconst}%
\end{figure}

Problem {\ProbA} differs than problem {\ProbB} in the average interference constraint. Thus the {\DOIC} is a lower bound on both the constrained and the unconstrained {\DOAC} as shown in Fig. \ref{Sum_Delay_Avg_DOAC_CSMA_Subopt_CSI_Unconst}. However, since the problem is delay limited and not interference limited, this delay increase is minor.% We note that the CSMA policy plotted in Fig. \ref{Sum_Delay_Avg_DOAC_CSMA_Subopt_CSI_Unconst} uses a ``genie-aided'' power allocation policy obtained from Algorithm \ref{DOACopt}. Thus, even when the two algorithms, the CSMA policy and the {\DOAC} policy, have the same power allocation policy, the {\DOAC} scheduling policy has an improved delay performance over the CSMA policy.

%\begin{figure}%
%\centering
%\includegraphics[width=\widthn\columnwidth]{figures/Avg_Delay_Avg_Inst_CSMA.pdf}%
%\caption{Comparing the CSMA policy with the {\DOIC} policy and the {\DOAC} policy. In this figure, the CSMA allocates the same power as the {\DOAC} policy, hence the  genie-aided terminology. The CSMA policy has a worse delay performance compared to the other two policies. The {\DOAC} yields higher delay than the {\DOIC} since it has an additional average interference constraint.}%
%\label{Avg_Delay_Avg_Inst_CSMA}%
%\end{figure}

%\begin{figure}%
%\centering
%\includegraphics[width=\widthn\columnwidth]{figures/Avg_Delay_Avg_vs_Sub.pdf}%
%\caption{The low-complexity algorithm proposed in Section \ref{Suboptimal} has a close-to-optimal average delay performance with a maximum error of $0.95\%$.}%
%\label{Avg_Delay_Avg_vs_Sub}%
%\end{figure}

%\begin{figure}%
%\centering
%\includegraphics[width=\widthn\columnwidth]{figures/Avg_Delay_Avg_CSI_Err.pdf}%
%\caption{Sum of cost functions for the perfect as well as the imperfect channel sensing for the {\DOAC} policy to solve the constrained optimization problem {\ProbB}.}
%\label{Avg_Delay_Avg_CSI_Err}%
%\end{figure}

\section{Conclusion}
\label{Conclusion}
We have studied the joint scheduling and power allocation problem of an uplink multi SU CR system. We formulated the problem as a delay minimization problem in the presence of average and instantaneous interference constraints to the PU, as well as an average delay constraint for each SU. Most of the existing literature that studies this problem either assume on-off fading channels or do not provide a delay-optimal algorithm which is essential for real-time applications.

We proposed a dynamic algorithm that schedules the SUs by dynamically updating a priority list based on the channel statistics, history of arrivals, departures and channel realizations. The proposed algorithm updates the priority list on a per-frame basis while controlling the power on a per-slot basis. %Based on these closed-form expressions for the power allocation and scheduling problem, the complexity of the proposed algorithm does not increase with the increase in the number of SUs.
We showed, through the Lyapunov optimization, that the proposed {\DOAC} policy is asymptotically delay optimal.% That is, it minimizes the sum of average delays of the SUs as well as satisfying the average interference and delay constraints.

When the number of SUs $N$ in the system is large, the complexity of the {\DOAC} policy scales as $O(MN\cdot 2^N)$, where $M$ is the number of iterations required to solve a one-dimensional search. Hence, we proposed a suboptimal algorithm with a complexity of $O(N\log(N))$ that does not sacrifice the performance significantly. Simulation results showed the robustness of the {\DOAC} policy against CSI estimation errors.

%We have studied the joint scheduling and power allocation problem of multi SUs in the presence of interference constraints to the PU, as well as an average delay constraint for each SU that needs to be met. We proposed a dynamic algorithm that schedules a SU each time slot based on the SUs' queues, channel gains, and interference channel gains. Moreover, the algorithm finds the power allocated to this user using water-filling-like expression. Based on these closed-form expressions for the power allocation and scheduling problem, the complexity of the proposed algorithm does not increase with the increase in the number of SUs. We showed, through the Lyapunov optimization, that the proposed algorithm is asymptotically delay optimal. That is, it minimizes the sum of the average delays of all SUs as well as satisfying the average interference and delay constraints.

\appendices
\section{Proof of Theorem \ref{Optimality}}
\label{Optimality_Proof_Inst}
\begin{proof}
In this proof, we show that the drift-plus-penalty under this algorithm is upper bounded by some constant, which indicates that the virtual queues are mean rate stable \cite{georgiadis2006resource,urgaonkar2011optimal}.

We define $\bfU(k)=\bfY(k)$ and the Lyapunov function as $L(k) \triangleq\frac{1}{2}\sum_{i=1}^N Y_i^2(k)$ and Lyapunov drift to be
\begin{equation}
\Delta (k) \triangleq \EEY{L(k+1) - L(k)},
\label{Drift_Def}
\end{equation}
Squaring \eqref{Delay_Q} then taking the conditional expectation we can write the following bound
\begin{multline}
\frac{1}{2}\E_{\bfU(k)} \left[ Y_i^2(k+1)-Y_i^2(k)\right] \leq \\
Y_i(k) \EEY{\FDurK}\lambda_i \lb \EEY{W_i^{(j)}}-r_i(k)\rb + C_{Y_i},
\label{Delay_Q_Sq2}
\end{multline}
where we use the bound $\EEY{\lb \sum_{j\in \script{A}_i(k)} W_i^{(j)}\rb^2}+\EEY{\lb\sum_{j\in \script{A}_(k)}r_i(k)\rb^2}<C_{Y_i}$. We omit the derivation of this bound due to lack of space. Given some fixed control parameter $V>0$, we add the penalty term $V\sum_i \EEY{r_i(k)\FDurK}$ to both sides of \eqref{Drift_Def}. Using the bound in \eqref{Delay_Q_Sq2} the drift-plus-penalty term becomes bounded by
\begin{equation}
\Delta \lb k\rb + V\sum_{i=1}^N \EEY{r_i(k)\FDurK}\leq C_Y+\EEY{\FDurK} \Phi_{\rm I}
\label{Drift_Plus_Penalty1}
\end{equation}
where $\Phi_{\rm I}$ is given by equation \eqref{Upper_Bound_Inst}. We define the {\DOIC} policy to be the policy that finds the values of $\bfpi(k)$, $\{\bfP^{(t)}\}$ and $\bfr(k)$ vector that minimize $\Phi_{\rm I}$ subject to the instantaneous interference, the maximum power and the single-SU-per-time-slot constraints in problem {\ProbA}. We can observe that the variables $\bfr(k)$, $\{\bfP^{(t)}\}$ and $\bfpi(k)$ can be chosen independently from each other. Step 4.a in the {\DOIC} policy finds the optimum value of $r_i(k)$, $\forall i\in\script{N}$. Moreover, since $\EEY{W_i^{(j)}}$ is decreasing in $\Pit$ $\forall t\in\script{F}(k)$, the optimum value for $\Pit$ is \eqref{Power_Allocation}. Finally, from \cite{c_mu_Rule} the $c\mu$-rule can be applied to find the optimum priority list $\bfpi(k)$ which is given by Step 1 in the {\DOIC} policy.
%\begin{equation}
%C_Y\triangleq \frac{1}{\sum_{i=1}^N\lambda_i}\sqrt{\frac{\sum_{i=1}^N \lambda_i\theta_i^{(4)}}{1-\sum_{i=1}^N\rho_i(P_i(k))}}
%\label{Busy_Per_Bound}
%\end{equation}

Now, since the proposed {\DOIC} policy minimizes $\Phi_{\rm I}$, this gives a lower bound on $\Phi_{\rm I}$ compared to any other policy including the optimal policy that solves {\ProbA}. Hence, we now evaluate $\Phi_{\rm I}$ at the optimal policy that solves {\ProbA} with the help of a genie-aided knowledge of $r_i(k)=\bW_i^*$ yielding $\Phi_{\rm I}^{\rm opt}=V\sum_{i=1}^N \bW_i^*$, where we use $\EEY{W_i^{(j)}}=\bW_i^*$. Substituting by $\Phi_{\rm I}^{\rm opt}$ in the right-hand-side (r.h.s.) of \eqref{Drift_Plus_Penalty1} gives an upper bound on the drift-plus-penalty when evaluated at the {\DOIC} policy. Namely
\begin{equation}
\Delta \lb k\rb + V\sum_{i=1}^N \EEY{r_i(k)\FDurK}\leq C_Y + V\sum_{i=1}^N \bW_i^*\EEY{\FDurK}.
\label{DOIC_Genie}
\end{equation}
Taking $\EE{\cdot}$, summing over $k=0,\cdots,K-1$, denoting $\bfY_i(0)\triangleq 0$ for all $i\in\script{N}$, and dividing by $V\sum_{k=0}^{K-1} \EE{\FDurK}$ we get
\begin{multline}
\sum_{i=1}^N \frac{\EE{Y_i^2(K)}}{\sum_{k=0}^{K-1} \EE{\FDurK}}+ \sum_{i=1}^N \frac{\sum_{k=0}^{K-1}\EE{r_i(k)\FDurK}}{\sum_{k=0}^{K-1}\EE{\FDurK}} \overset{(a)}{\leq}\\
\frac{aC_Y}{V} + \sum_{i=1}^N \bW_i^*\triangleq C_1.
\label{Optimal_Eq}
\end{multline}
where in the r.h.s. of inequality (a) we used $\EE{\FDurK}\geq \EE{I(k)}=1/a$, and $C_1$ is some constant that is not a function in $K$. To prove the mean rate stability of the sequence $\{Y_i(k)\}_{k=0}^\infty$ for any $i\in\script{N}$, we remove the first and third terms in the left-side of \eqref{Optimal_Eq} as well as the summation operator from the second term to obtain $\EE{Y_i^2(K)}/K \leq C_1$ $\forall i\in\script{N}$. Using Jensen's inequality we note that $\EE{Y_i(K)}/K \leq \sqrt{\EE{Y_i^2(K)}/K^2} \leq \sqrt{C_1/K}$.
%\begin{equation}
%\frac{\EE{Y_i(K)}}{K} \leq \sqrt{\frac{\EE{Y_i^2(K)}}{K^2}} \leq \sqrt{\frac{C_1}{K}}.
%\label{Jensens}
%\end{equation}
Finally, taking the limit when $K\rightarrow \infty$ completes the mean rate stability proof. On the other hand, to prove the upper bound in Theorem \ref{Optimality}, we use the fact that $r_i(k)$ and $\vert \script{A}_i(k) \vert$ are independent random variables (see step 4-a in {\DOIC}) to replace $\EE{\vert \script{A}_i(k) \vert {r_i(k)}}$ by $\lambda_i\EE{\FDurK r_i(k)}$ in \eqref{Wait_r_i}, then we take the limit of \eqref{Wait_r_i} as $K\rightarrow \infty$, use the mean rate stability theorem and sum over $i\in\script{N}$ to get
\begin{multline}
\sum_{i=1}^N \frac{\EE{\sum_{k=0}^{K-1} \lb\sum_{j\in \script{A}_i(k)}W_i^{(j)}\rb}}{\EE{\sum_{k=0}^{K-1}{\vert\script{A}_i(k)\vert}}} \leq \sum_{i=1}^N \frac{\sum_{k=0}^{K-1}\EE{r_i(k)\FDurK}}{\sum_{k=0}^{K-1}\EE{\FDurK}}\\
\overset{(b)}{\leq} \frac{aC_Y}{V} + \sum_{i=1}^N \bW_i^*,
\label{Optimality_Eq2}
\end{multline}
where inequality (b) comes from removing the first summation in the left-side of \eqref{Optimal_Eq}. Taking the limit when $K\rightarrow \infty$ and using \eqref{Delay_Frame} completes the proof.
\end{proof}

\section{Existence of The Service Time Moments}
\label{No_Deep_Fade}
\begin{lma}
\label{No_Deep_Fade_Lemma}
Given any distribution for $\Pit\gamma_i^{(t)}$ the inequality $\EE{s_i^n}<\infty$ holds $\forall n\geq 1$.% Moreover, when the power is given by $\Pit=\min\lb \Iinst/\git,P\rb$ for some fixed parameter $P\in[\Pminn_i,\Pmax]$, the inequality $\EE{s_i^2}\leq \lb L^2+L\lb 1-p_i(\Pminn_i)\rb\rb/p_i^2(\Pminn_i)$ holds with $p_i(P)\triangleq 1-\Prob{R_i(P)=0}$.
\end{lma}
\begin{proof}
Given some, possibly random, power allocation policy $\Pit$ let the random variable $\sB_i\triangleq \sNB_i+L$ where $\sNB_i$ is a random variable following the negative binomial distribution \cite[pp. 297]{degroot2011probability} with success probability $1-\Prob{\Ri=0}$ while number of successes equals $L$. We can show that $\Prob{s_i\leq x}\geq\Prob{\sB_i\leq x}$. Hence, according to the theory of stochastic ordering, the moments of $s_i$ are upper bounded by their respective moments of $\sB_i$ \cite[equation (2.14) pp. 16]{rajan2014ordering}. The lemma holds since all the moments of $\sB_i$ exist, a fact that is based on the fact that the moments of the negative binomial distribution exist \cite[pp. 297]{degroot2011probability}.

%For the second part of the lemma, we set $\Pit=\min\lb \Iinst/\git,P\rb$ for some deterministic parameter $P\geq\Pminn_i$ and define $p_i(P)\triangleq 1-\Prob{R_i(P)=0}$. We can show that the first two moments of $\sNB_i$ %from its moment generating function $p_i^L(P)/\lb1- \lb 1-p_i(P)e^x\rb\rb^L$ as
%%\begin{align}
%%\EE{\sNB_i}&=\frac{\lb1-p_i(P)\rb L}{p_i(P)}\\
%%\EE{\lb\sNB_i\rb^2}&=\frac{\lb 1-p_i(P)\rb^2L^2+\lb1-p_i(P)\rb L}{p_i^2(P)},
%%\label{NB_Two_Moments}
%%\end{align}
%are decreasing in $p_i(P)$. The proof follows using the bound $p_i(P)\geq p_i(\Pminn_i)$ and the inequality $\EE{s_i^2}\leq\EE{\lb\sB_i\rb^2}=\EE{\lb\sNB_i\rb^2}+2L\EE{\sNB_i}+L^2$.
\end{proof}

\section{Proof of Theorem \ref{Optimality_Avg}}
\label{Optimality_Proof}
\begin{proof}
This proof is similar to that in Appendix \ref{Optimality_Proof_Inst}. We define $\bfU(k)\triangleq [X(k) , \bfY(k)]^T$, the Lyapunov function as $L(k) \triangleq \frac{1}{2}X^2(k)+\frac{1}{2}\sum_{i=1}^N Y_i^2(k)$ and Lyapunov drift as in \eqref{Drift_Def}. Following similar steps as in Appendix \ref{Optimality_Proof_Inst} and using the bound $\EEU{\lb\sum_{i=1}^N\sum_{t\in\script{F}(k)}\Pit \git\rb^2+\lb\Iavg \FDurK\rb^2}<C_X$, where $C_X$ is defined before Theorem \ref{Optimality_Avg}, we get the following bound on the drift-plus-penalty term
%\begin{equation}
%\Delta (k) \triangleq \EEU{L(k+1) - L(k)}.
%\label{Drift_Def_Avg}
%\end{equation}
%Squaring \eqref{Avg_Interf_Q} then taking the conditional expectation we can get the bound
%\begin{equation}
%\frac{\E_{\bfU(k)} \left[X^2(k+1)-X^2(k)\right]}{2} \leq C_X+X(k)\lb\EEU{\sum_{t\in\script{F}(k)}\Pit \git}-\Iavg\EEU{\FDurK}\rb,
%\label{Interf_Q_Sq1}
%\end{equation}
%where we use the bound $\EEU{\lb\sum_{i=1}^N\sum_{t\in\script{F}(k)}\Pit \git\rb^2+\lb\Iavg \FDurK\rb^2}<C_X$ in \eqref{Interf_Q_Sq1} and omit the derivation of this bound. Given some fixed control parameter $V>0$, we add the penalty term $V\sum_i \EEU{r_i(k)\FDurK}$ to both sides of \eqref{Drift_Def_Avg}. Using the bounds in \eqref{Delay_Q_Sq2} and \eqref{Interf_Q_Sq1}, the drift-plus-penalty term becomes bounded by
\begin{equation}
\Delta \lb \bfU(k)\rb + V\sum_{i=1}^N \EEU{r_i(k)\FDurK}\leq C+\EEU{\FDurK}\chi(k),
\label{Drift_Plus_Penalty_Avg}
\end{equation}
where
\begin{equation}
\chi(k)\triangleq \sum_{i=1}^N \lb V -Y_i(k) \lambda_i\rb r_i(k)+\Phi_{\rm A},
\label{chi}
\end{equation}
with
\begin{multline}
\Phi_{\rm A}\triangleq\sum_{l=1}^N \lb Y_{\pi_l}(k) \lambda_{\pi_l} \EEU{W_{\pi_l}^{(j)}} + \right.\\
\left. X(k)\lb\frac{\EEU{\sum_{t\in\script{F}(k)}P_{\pi_l}^{(t)} g_{\pi_l}^{(t)}}}{\EEU{\FDurK}} -\Iavg\rb\rb
\end{multline}
We define the {\DOAC} policy to be the policy that jointly finds $\bfr(k)$, $\{\bfP^{(t)}\}$ and $\bfpi(k)$ that minimize $\chi(k)$ subject to the instantaneous interference, the maximum power and the single-SU-per-time-slot constraints in problem {\ProbB}. Step 5-a in the {\DOAC} policy minimizes the first summation of $\chi(k)$. For $\{\bfP^{(t)}\}$ and $\bfpi(k)$, we can see that $\Phi_{\rm A}$ is the only term in the right side of \eqref{chi} that is a function of the power allocation policy $\{\bfP^{(t)}\}$, $\forall t\in\script{F}(k)$. For a fixed priority list $\bfpi(k)$, using the Lagrange optimization to find the optimum power allocation policy that minimizes $\Phi_{\rm A}$ subject to the aforementioned constraints yields \eqref{Pow_Allocation}, where $P_{\pi_j}(k)$, $\forall i\in\script{N}$, is some fixed power parameter that minimizes $\Phi_{\rm A}$ subject to the maximum power constraint only. Substituting by \eqref{Pow_Allocation} in $\Phi_{\rm A}$ and using the bound $\EEU{W_{\pi_l}^{(j)}}=W_{\pi_l}(P_{\pi_l}(k))\leq W_{\pi_l}^{\rm up}(P_{\pi_l}(k))$ we get $\Psi$ that is defined before \eqref{Optimization_Obj}. Consequently, $\bfPst(k)$ and $\bfpist(k)$, the optimum values for $\bfP(k)$ and $\bfpi(k)$ respectively, are the ones that minimize $\Psi$ as given by Algorithm \ref{DOACopt}.

Since the optimum policy that solves {\ProbB} satisfies the interference constraint, i.e. satisfies $\EEU{\sum_{t\in\script{F}(k)}P_{\pi_l}^{(t)} g_{\pi_l}^{(t)}} \leq\EEU{\FDurK}\Iavg$, we can evaluate $\chi(k)$ at this optimum policy with a genie-aided knowledge of $r_i(k)=\bW_i^*$ to get $\chi^{\rm opt}\triangleq V\sum_{i=1}^N\bW_i^*$. Replacing $\chi(k)$ with $\chi^{\rm opt}$ in the r.h.s. of \eqref{Drift_Plus_Penalty_Avg} we get the bound $\Delta \lb \bfU(k)\rb + V\sum_{i=1}^N \EEU{r_i(k)\FDurK}\leq C+\EEU{\FDurK}V\sum_{i=1}^N\bW_i^*$. Taking $\EE{\cdot}$ over this inequality, summing over $k=0,\cdots,K-1$, denoting $X(0)\triangleq \bfY_i(0)\triangleq 0$ for all $i\in\script{N}$, and dividing by $V\sum_{k=0}^{K-1} \EE{\FDurK}$ we get
\begin{multline}
\frac{\EE{X^2(K)}}{\sum_{k=0}^{K-1}\EE{\FDurK}}+\sum_{i=1}^N \frac{\EE{Y_i^2(K)}}{\sum_{k=0}^{K-1} \EE{\FDurK}}+ \sum_{i=1}^N \frac{\sum_{k=0}^{K-1}\EE{r_i(k)\FDurK}}{\sum_{k=0}^{K-1}\EE{\FDurK}} \\
\leq \frac{CK}{V\sum_{k=0}^{K-1}\EE{\FDurK}} + \sum_{i=1}^N \bW_i^*.
\label{Optimal_Eq_Avg}
\end{multline}
Similar steps to those in Appendix \ref{Optimality_Proof_Inst} can be followed to prove the mean rate stability of $\{X(k)\}_{k=0}^\infty$ and $\{Y_i(k)\}_{k=0}^\infty$ as well as the bound in Theorem \ref{Optimality_Avg}, and thus are omitted here.
\end{proof}

\bibliographystyle{IEEEbib}
\bibliography{MyLib}


\begin{IEEEbiography}[{\includegraphics[width=1in,height=1.25in,clip,keepaspectratio]{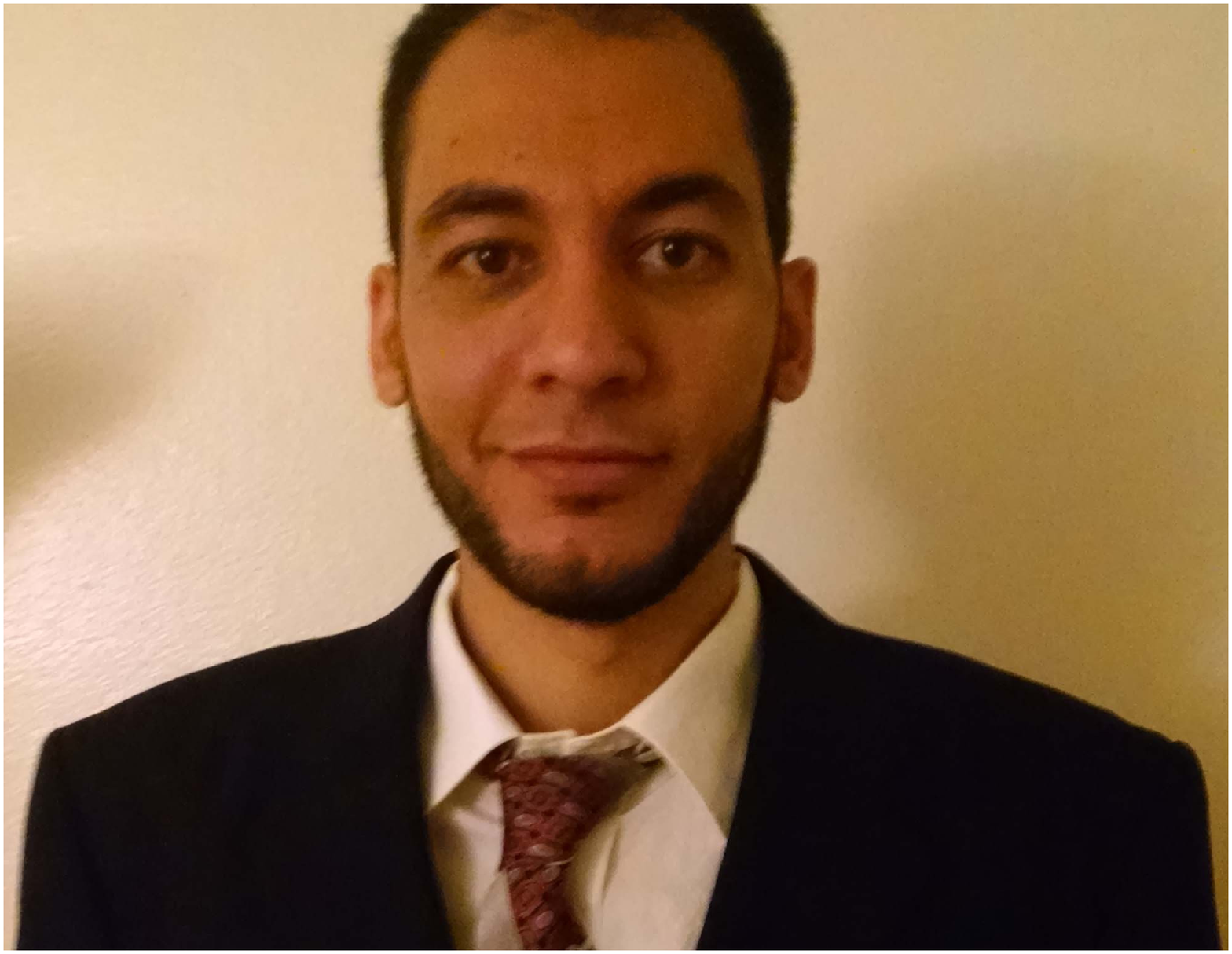}}]{Ahmed E. Ewaisha} was born in Cairo, Egypt in 1987. He received his B.S. degree with honors in electrical engineering ranking top 5\% on his class at Alexandria University in 2009. Consequently, he was admitted to Nile University that is considered the first research-based university in Egypt where he received his M.S. degree in 2011 in wireless communications.

In fall 2011, he joined the Ira A. Fulton School of engineering at Arizona State University, Tempe, where he is now working towards his PhD degree studying the delay analysis in cognitive radio networks. His research interests span a wide area of wireless as well as wired communication networks including stochastic optimization, power allocation, cognitive radio networks, resource allocation and quality-of-service guarantees in data networks.
\end{IEEEbiography}

\begin{IEEEbiography}[{\includegraphics[width=1in,height=1.25in,clip,keepaspectratio]{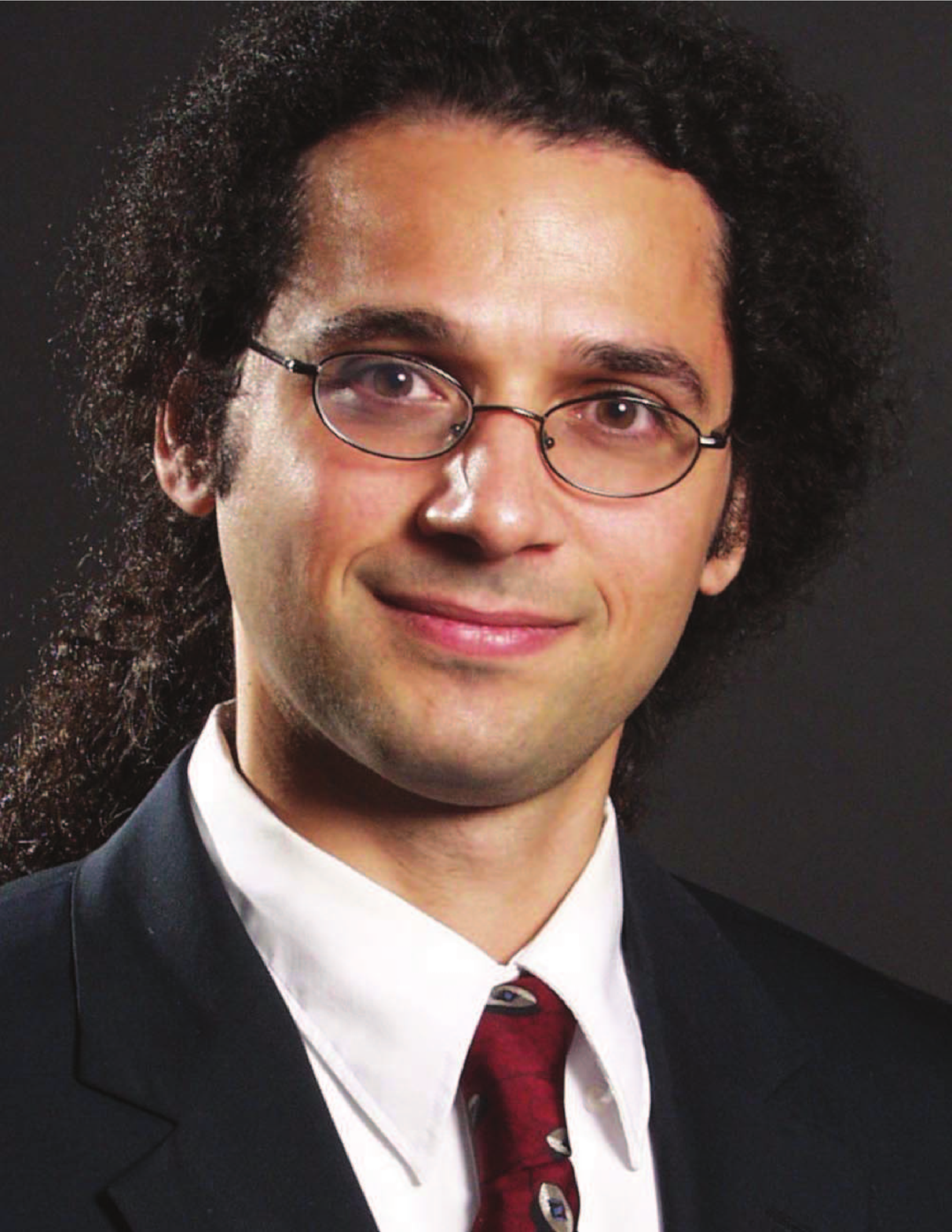}}]{Cihan Tepedelenlio\u{g}lu} (S’97-M’01) was born in Ankara, Turkey in 1973. He received his B.S. degree with highest honors from Florida Institute of Technology in 1995, and his M.S. degree from the University of Virginia in 1998, both in electrical engineering. From January 1999 to May 2001 he was a research assistant at the University of Minnesota, where he completed his Ph.D.  degree in electrical and computer engineering. He is currently an Associate Professor of Electrical Engineering at Arizona State University.

Prof. Tepedelenlio\u{g}lu was awarded the NSF (early) Career grant in 2001, and has served as an Associate Editor for several IEEE Transactions including IEEE Transactions on Communications, IEEE Signal Processing Letters, and IEEE Transactions on Vehicular Technology. His research interests include statistical signal processing, system identification, wireless communications, estimation and equalization algorithms for wireless systems, multi-antenna communications, OFDM, ultra-wideband systems, distributed detection and estimation, and data mining for PV systems.
\end{IEEEbiography}

%\newpage



%\vfill


\end{document}